\documentclass[aip,jmp,graphicx]{revtex4-1}
\usepackage{amsmath,amssymb,amsthm}
\usepackage{graphicx}

\newcommand{\be}{\begin{equation}}
\newcommand{\ee}{\end{equation}}
\newcommand{\bea}{\begin{array}}
\newcommand{\ea}{\end{array}}
\newcommand{\beqa}{\begin{eqnarray}}
\newcommand{\eeqa}{\end{eqnarray}}
\newcommand{\bean}{\begin{eqnarray*}}
\newcommand{\eean}{\end{eqnarray*}}

\numberwithin{equation}{section}

\newtheorem{proposition}{\textsc{Proposition}}[section]

\theoremstyle{definition}

\theoremstyle{remark}

\begin{document}


\title{Induced Gravity from Quantum Spacetime by Twisted Deformation of the Quantum Poincar\'e Group } 



\author{Cesar A. Aguill\'on}
\email[]{ceaguillon@ciencias.unam.mx}
\affiliation{Instituto de Ciencias Nucleares,\\
Universidad Nacional Aut\'onoma de M\'exico,\\
 A. Postal 70-543 , M\'exico D.F., M\'exico \\
}
\author{Albert Much}
\email[]{amuch@matmor.unam.mx}
\affiliation{Instituto de Ciencias Nucleares,\\
Universidad Nacional Aut\'onoma de M\'exico,\\
 A. Postal 70-543 , M\'exico D.F., M\'exico \\
}
\affiliation{Centro de Ciencias Matem\'aticas, \\
 UNAM, Morelia, Michoac\'an, M\'exico
 }

\author{Marcos Rosenbaum}
\email[]{mrosen@nucleares.unam.mx}
\affiliation{Instituto de Ciencias Nucleares,\\
Universidad Nacional Aut\'onoma de M\'exico,\\
 A. Postal 70-543 , M\'exico D.F., M\'exico \\
}

\author{J. David Vergara}
\email[]{vergara@nucleares.unam.mx}
\affiliation{Instituto de Ciencias Nucleares,\\
Universidad Nacional Aut\'onoma de M\'exico,\\
 A. Postal 70-543 , M\'exico D.F., M\'exico \\
}


\date{\today}

\begin{abstract}
We investigate a quantum geometric space in the context of what could be considered an emerging effective theory from Quantum Gravity. Specifically we consider a two-parameter class of twisted Poincar\'e algebras, from which Lie-algebraic noncommutativities of the translations are derived as well as associative star-products, deformed Riemannian geometries, Lie-algebraic twisted Minkowski spaces and quantum effects that arise as noncommutativities. Starting from a universal differential algebra of forms based on the above mentioned Lie-algebraic noncommutativities of the translations, we construct the noncommutative differential forms and Inner and Outer derivations, which are the noncommutative equivalents of the vector fields in the case of commutative differential geometry. Having established the essentials of this formalism we construct a bimodule, required to be central under the action of the Inner derivations in order to have well defined contractions and from where the algebraic dependence of its coefficients is derived. This again then defines the noncommutative equivalent of the geometrical line-element in commutative differential geometry. We stress, however, that even though the components of the twisted metric are by construction symmetric in their algebra valuation, this is not so for their inverse and thus to construct it we made use of Gel'fand's theory of quasi-determinants, which is conceptually straightforward but computationally becoming quite complicate beyond an algebra of 3 generators. The consequences of the noncommutativity of the Lie-algebra twisted 
geometry are further discussed.

\end{abstract}


\maketitle 

\section{Introduction}
Reductionism  is an essential concept in Physics and in this spirit the Quantization of Relativity, at distances of the order of the Planck length, still remains to be one of the main great problems in the field. It is then not surprising that there is a great deal of work going on in this area. Thus, for instance, in \cite{CCM} it is proposed that a quantum geometry appears from a higher degree version of the Heisenberg commutation relations. On the other hand in \cite{BM} a theory of Poisson-Riemannian geometry an analysis is introduced  on the constraints on the classical geometry in such a way that the quantization exists. A natural path to define a noncommutative geometry is to extend the Gel'fand-Naimark-Segal duality to noncommutative algebras, where the noncommutative geometry of spacetime can be described by using an associative algebra $\mathcal A$, with unit, and the derivations $Der({\mathcal A})$ 
of $\mathcal A$ are graded differential Lie algebras and also $\mathcal A$-modules, with $Der({\mathcal A})$ playing the role of the Lie algebra vector fields \cite{dubois1}. This approach has appeared in several publications, for a review of it see for example \cite{dubois5}. An extension of this work, consisting in considering that the metric is a two-form central to the algebra  $\mathcal A$ implies that the metric is fixed, up to some free parameters \cite{b&m}.

As pointed out in \cite{dop}, in the case of Quantum Gravity the 
conjunction of the principles of Quantum Mechanics and Classical General Relativity imposes limits on the joint precision allowed in the measurement of the space-time coordinates of an event, due to the fact that the concentrated energy required by the Heisenberg Uncertainty Principle in order to localize an event should not be so strong as to hide the event itself to any distant observer - distant compared to the Planck length scale. These limitations lead to the space-time uncertainty relations
\be\label{1.1}
\Delta x^0 \sum^3 _{i=1} \Delta x^i \geq \lambda^2_P \;\;\;\;\; \sum_{j\leq 1 <k\leq 3} \Delta x^j \Delta x^k \geq \lambda^2_P,
\ee
 It has also been shown \cite{fred, rob} that the above uncertainty relations were exactly implemented by Commutation Relations of the form
\be\label{uncert}
[\hat{x}^\mu , \hat{x}^\nu ]=\frac{i}{\kappa^2}\; \theta^{\mu\nu}(\kappa \hat x),
\ee
where $\kappa$ is identified with the inverse of the Planck length  $\lambda_P= \large(\frac{G\hbar}{c^3}\large)^{1/2}$, and the limit in this case corresponds to $\zeta\to 0$, or $\kappa\to \infty$, while
\be\label{1.3}
\theta^{\mu\nu}(\kappa\hat x) = \theta^{\mu\nu}_{(0)} + \kappa \theta^{\mu\nu\rho}_{(1)} \hat x_\rho + \kappa^2 \theta^{\mu\nu\rho\tau}_{(2)} \hat x_\rho \hat x_\tau +\dots,
\ee
with $\theta^{\mu\nu}$  an antisymmetric tensor of units $\lambda^2_P$ (which can be set equal to $1$ by adopting the absolute units $\hbar=c=G=1$). Under certain requirements such as Lorentz invariance, this tensor has to satisfy certain ``quantum conditions" that imply that the Euclidean distance operator has a spectrum bounded from below by a constant of the order of the Planck length. However, if the classical Poincar\'e symmetries are not modified  then the numerical constant tensors 
$\theta_{(n)}^{\;\;\mu\nu\rho_1\dots \rho_n}$ break the Lorentz invariance, so a quantum deformation is needed in order that the commutator (\ref{uncert}) remains invariant \cite{l-w}.

Other models  lead to different space-time uncertainties and involve limitations that
do not pose restrictions in the measurement of a single coordinate, although they suggest, however, minimal uncertainties in the measurements of area and volume operators.
Thus, even though the various different models so far considered in the literature (see {\it e.g.} \cite{camelia, tomassi, martinetti}  and references therein) lead to different quantum conditions that the tensor $\theta^{\mu\nu}$ is required to satisfy, they all point out to the concept that, due to limitations in localizability for events below the Planck scale, space-time rather than appearing as a smooth manifold is expected to be more appropriately described as a mathematical object (the quantum space), where ``coordinates" are self-adjoint operators acting on some Hilbert space, such that the spectrum of space-time observables, constructed from them, is bounded from below by dimensions in the orders of powers of the Planck length. So, from a qualitative and operational meaningful point of view, their common denominator suggests a sort of discrete cellular structure for describing physical space \cite{connes1}, \cite{landi}.\\
In particular, a Physical-Mathematical directrix that summarizes the great advances in unification of the Fundamental Interactions are the Yang-Mills and Gravitation theories, founded on the notion of a connection (gauge or linear) on fiber bundles. This has opened the possibility to extend these notions to the field of Noncommutativity Geometry, based on a new classical duality: the Theorem of Serre-Swan \cite{S-S} that establishes a complete equivalence between the category of vector bundles over a smooth and compact space and its maps, with the category of projective modules of a finite type over commutative algebras and their module morphisms. The above, together with the Gel'fand-Naimark Theorem which states that, given any commutative $C^{\star}$-algebra, a Hausdorff topological space can be reconstructed so that the algebra can be isometrically and $\star$-isomorphically reconstructed as an algebra of complex valued continuous functions and therefore implies that, for a noncommutative algebra as a starting ingredient, the analogy to vector bundles is the projective modules of a finite type over that algebra  \cite{G-N},  \cite{Dor},  \cite{Fuji}. 

Thus, even though a complete and final Quantum Gravity is still beyond reach, it is reasonable to expect that, as an Ontologically more fundamental category, Noncommutative Geometry constitutes a promising effective theory  approach to  Quantum Gravity at distances of an order of the Planck length,
whereby, in some emerging limit of the full theory, General Relativity would be recovered when $\kappa\to\infty$. Such an effective theory would be Quantum Geometry, as it seems naturally suggested by the arguments above which lead to the hypothesis that quantum space-time is an effect of Quantum Gravity. Such a theory would of course be a bridge between the possibly combinatorial Quantum Gravity and the geometry of the classical continuum which should be obtained in some limit. Hence, for the present, we must consider Noncommutative Geometry  as a more general notion that, because of its noncommutativity, it must be a correct scenario for the phenomenology and test of the first quasi-classical corrections to quantum gravity, but beyond that, in the realm of strong emergence, it would indicate the mathematical constrictions on the structure of quantum gravity itself that would have to emerge naturally from the complete and final theory \cite {b&m}.\

The plan of the paper is as follows. In Section 2 we introduce our spacetime algebra ${\mathcal A}$ as a Lie-algebraic deformation of the Minkowski spacetime, based on the work in \cite{l-w}. In Section 3, to fix notation and for self-consistency and we present a brief summary of the elements of noncommutative differential forms and derivation algebras on which the paper is based and derive the quantum metric components required in order to have a central metric bimodule; we also compute the quantum determinant and the quantum inverse metric based on the work in \cite{GG}. In Sections 4 and 5, we apply the previously outlined elements on the theory of Inner and Outer derivations as well as of noncommutative differential geometry to derive the covariant derivative and noncommutative connection symbols, in general. We then consider the case of no-torsion and metricity that, although contrary to what happens in the classical Levi-Civita differential geometry, the connection is not unique but allows us to establish further relations between the connection symbols in order to provide a somewhat more tractable noncommutative Riemannian Geometry for the metric central bimodule and a more computationally viable approach to the inverse  matrix presented in Section 3, which could be implemented in the contractions of the noncommutative Riemannian geometry. We conclude finally with some observations concerning the effect of the noncommutativity on the twisted geometry and give an outlook on quantization by the addition of a Hermitian structure to the formalism, as well as  Appendices A and B where some of the more extensive details needed are included in order to make the derivations of the twisted determinant and metricity in the text more tractable.
 \\
 \section{Quantum Poincar\'e Algebras and Quantum Poincar\'e Groups}   
 
Now, in order to construct a plausible Quantum Geometry we first need to construct a quantum isometry that leaves invariant a deformed metric. For this purpose we need to 
consider first, as mentioned in the Introduction, the appropriate quantum deformations of Poincar\'e symmetries such that Lorentz invariance of the commutator (\ref{uncert}) is preserved. Making use of the results in \cite{l-w} and \cite{z}, which we summarize here, we have that an arbitrary four-dimensional Poincar\'e algebra of r-matrices can be split as
\be\label{2.4}
r= \frac{1}{2\kappa^{2}} \theta^{\mu\nu}_{(0)} P_\mu \wedge P_\nu + \frac{1}{2\kappa} \theta^{\mu\nu\rho}_{(1)} P_\mu \wedge M_{\nu\rho}
+\frac{1}{2}\theta^{\mu\nu\rho\sigma}_{(2)} M_{\mu\nu}\wedge M_{\rho\sigma}.
\ee
However, since the twist-deformed Hopf Poincar\'e algebra, generated by the Abelian carrier algebra $[P_\mu , P_\nu]=0$, does not agree with the translation sector of 
the dual $ \theta^{\mu\nu}_{(0)}$ deformed - Poincar\'e group, we shall consider in this work the case of the Lie-algebraic deformation where 
\be\label{2.5)}
\theta^{\mu\nu\rho}_{(1)}=\epsilon^{\mu\nu\rho\tau}v_\tau, \quad \theta^{\mu\nu}_{(0)}=\theta^{\mu\nu\rho\sigma}_{(2)}=0,
\ee
and the indices $\nu,\rho$ are fixed while $v_{\tau}$ is a numerical four-vector with two vanishing components associated with the $\nu,\rho$ indices. 
Here the r-matrix describes an Abelian deformation with carrier algebra described by the generators $M_{\alpha\beta}, P_\lambda$.

Using the $5\times 5$-matrix representation for the Poincar\'e generators:
\be\label{matrep}
(M_{\mu\nu})^a_{\;b}=\delta^a_\mu \eta_{\nu b} - \delta^a_\nu \eta_{\mu b}, \quad\quad (P_\mu)^a_{\; b}=\delta^a_\mu \delta^0_b, 
\ee
where $\eta$ stands for the Minkowski metric tensor, it can be shown that in the universal matrix \\
$\mathcal R_{(1)}={\mathcal F}_{\zeta}^{T}{\mathcal F}_{\zeta}^{-1}$, 
where
\be\label{twistzeta}
{\mathcal F}_{\zeta}^{-1}=\exp\left(- \frac{i}{2\kappa}(\zeta^{\lambda} P_{\lambda} \wedge M_{\alpha\beta})\right)
\ee
is the twist function of the Hopf algebra $\mathcal U_\zeta (\mathcal P)$, with $\alpha\not=\beta$ and fixed. In the above expression only the term linear in $\frac{i}{\kappa}$ is non-vanishing. 
Thus
\be\label{univr}
\mathcal R_{(1)}= 1\otimes 1 - \frac{i}{\kappa}\left(\zeta^{\lambda} P_{\lambda} \wedge M_{\alpha\beta}\right)
\ee
 Moreover, letting
\be\label{pqg}
\mathcal R_{(1)}{\mathcal T}_1 {\mathcal T}_2 = {\mathcal T}_2 {\mathcal T}_1 \mathcal R_{(1)},
\ee
with
\be\label{repp}
(\mathcal T)^a_{\; b}=\begin{pmatrix}\hat\Lambda^{\mu}_{\;\nu}& \hat a^\mu\\
                       0& 1\end{pmatrix},\ee
where ${\mathcal T}_1={\mathcal T}\otimes1$, ${\mathcal T}_2=1\otimes{\mathcal T}$, $\hat\Lambda^{\mu}_{\;\nu}$ parameterizes the quantum Lorentz rotation and $\hat a^\mu$ denotes quantum translations,
yields the $\zeta^\lambda$-deformed matrix Poincar\'e quantum group ${\mathcal G}_{\zeta}$.\\
The translation sector of this group, resulting from the above relations, is given by 
\be\label{trans}
[\hat a_\mu, \hat a_\nu]= \frac{i}{\kappa}\zeta_\nu (\eta_{\mu\alpha}\;\hat a_\beta -\eta_{\mu\beta}\;\hat a_\alpha ) +\frac{i}{\kappa}\zeta_\mu (\eta_{\nu\beta}\;\hat a_\alpha -\eta_{\nu\alpha}\;\hat a_\beta ),
\ee
 with the $\alpha$,  $\beta$ indices fixed and the vector $\boldsymbol\zeta$ has vanishing components $\zeta^{\alpha},\zeta^{\beta} $.

Consider now the Lie-algebraic deformation ${\mathcal M}_\zeta$ of the Minkowski space when resorting to the Hopf module algebra $\mathcal U_{\zeta}(\mathcal P)$, dual to ${\mathcal G}_{\zeta}$. 
We have that
\be\label{mink}
f({\bf x})\star g({\bf x}) := m_{\mathcal F} \circ (f({\bf x})\otimes g({\bf x}) )= m\circ\left({\mathcal F}^{(-1)}\triangleright f({\bf x})\otimes g({\bf x})\right ),
\ee
in general, and for the twist (\ref{twistzeta}):

\be\label{twistzeta2}
f({\bf x})\star g({\bf x})=m\circ\left(\exp(- \frac{i}{2\kappa}(\zeta^{\lambda} P_{\lambda} \wedge M_{\alpha\beta})f({\bf x})\otimes g({\bf x}) \right),
\ee
where the Schwartz functions $f({\bf x})$ and $g({\bf x})$ are Weyl symbols. Specifically, the resulting commutator for the $\zeta$-deformed Minkowski space-time coordinates is:
 \be\label{opdef1}
[ x_{\mu},  x_{\nu}]_\star = {\mathcal C}_{\mu\nu}^{\rho}  x_{\rho},
\ee
with 
\be\label{const1}
{\mathcal C}_{\mu\nu}^{\rho}= \frac{i}{\kappa}\zeta_{\mu}(\eta_{\nu\beta} \delta^{\rho}_{\alpha} -\eta_{\nu\alpha}\delta^{\rho}_{\beta}) + 
\frac{i}{\kappa}\zeta_{\nu}(\eta_{\mu\alpha}\delta^{\rho}_{\beta}-\eta_{\mu\beta}\delta^{\rho}_{\alpha}),
\ee
where indices of the structure constants are lowered or raised with the flat Minkowski metric $\eta_{\mu\nu}$.
In addition, the fact that the $\alpha, \beta$ components of the vector $\boldsymbol \xi$ are zero in the structure constants (\ref{const1}), allows us to re-express the space-time algebra (\ref{opdef1}) as :
\be\label{sub5}
[ X_\alpha ,  X_\lambda ] = 2i\xi_{\lambda}\eta_{\alpha\alpha} X_\beta, \quad [ X_\beta ,X_\lambda ] = -2i\xi_{\lambda}\eta_{\beta\beta} X_\alpha,
\ee

To simplify the algebra, we will analyze the case of $\alpha=1, \ \ \beta=2$. Note first that since $X_0 ,X_3$ and $X_1 ,X_2$ commute among themselves, we can immediately write the Lie algebra as
\setlength{\parindent}{0pt}
	\begin{align}\label{lie}
	[X_{1},X_{\lambda}]&=-2i\xi_{\lambda} X_{2},\nonumber\\
	[X_{2},X_{\lambda}]&=2i\xi_{\lambda}X_{1},
	\end{align}
where $\lambda=0,3$. Namely
	\begin{align*}
	[X_1,X_{0}]&=-2i\xi_{0}{X}_{2}, \qquad [{X}_{1},{X}_{2}]=0,\\
	[{X}_{1},{X}_{3}]&=-2i\xi_{3}{X}_{2},\qquad[{X}_{2},{X}_{0}]=2i\xi_{0}{X}_{1},\\
	[{X}_{2},{X}_{3}]&=\;\;\;2i\xi_{3}{X}_{1}, \qquad[{X}_{0},
{X}_{3}]=0,
	\end{align*}
with the center elements of the algebra ${\mathcal O_1}, {\mathcal O_2 \in {\mathcal Z(\mathcal A)}}$ given by 
\be\label{cas1}
{\mathcal O_1}:= \xi_3  X_0 -\xi_0 X_3, \qquad  {\mathcal O_2}:= ( X_1)^2 +( X_2)^2.
\ee

 It is also interesting and noteworthy to point out an isomorphism of the reduced algebra.

  \begin{proposition}
  	The reduced algebra generated by ${X_1}$, ${X_2}$  and ${X_3}$, i.e. 
  	    \begin{align*}
  	    [{X}_1 ,{X}_2 ]=0,\quad  [{X}_1 ,{X}_3 ]=-2i\xi_{3}{X}_2,\qquad     [{X}_2 ,{X}_3 ]= 2i\xi_{3}{X}_1   
  	    \end{align*}  
  	  is isomorphic   to the two-dimensional Poincare group, where the isomorphism is given by the following identification 
  	 \begin{align}\label{iso}
  	 {X}_1= i\beta P_{0},\qquad   {X}_2=  \beta P_{1}, \qquad   {X}_3= 2i\xi_{3}M_{01}.
  	 \end{align}    	
  	\end{proposition}

  \begin{proof}
  	The proof is straightforward by using the commutator relations 
  	 \begin{align*}
  	 [P_{0},P_{1}]=0,\qquad [P_{0},M_{01}]=iP_{1},\qquad [P_{1},M_{01}]=iP_{0},
  	 \end{align*}           
  	 and the identifications  (\ref{iso}).        
  \end{proof}	

 Thus, by comparing equation (\ref{opdef1}) with (\ref{trans}), we note that the deformed Minkowski space can be identified with the translation sector of the Poincar\'e quantum group $\mathcal G_\zeta$.
We shall make use of this identification in the next section.

\section{The  quantum metric as the center of the deformed Poincar\'e Lie algebra and its inverse}

\subsection{The Graded Universal Differential Forms and Derivation Algebras}
As pointed out in the Introduction, it is well known that the Serre-Swan theorem \cite{S-S} establishes a central link between projective finite modules and vector bundles so that an equivalence of categories exists between the vector bundle and the projective module of its smooth sections. 
It is thus pertinent to consider that the natural generalization of the notion of a vector bundle, which is an essential ingredient in the formulation of field theories, should be a finite projective module. Hence, having an appropriate 
noncommutative generalization of the algebra of differential forms $\Omega^n$ there is a natural connection on modules which generalizes the notion of connection on vector bundles \cite{dubois1}, \cite{connes1}, \cite{landi},\cite{connes2}, \cite{dubois2}, \cite{D-V4}. \\
Briefly, let $\Omega(\mathcal U (A))= \bigoplus_{n} \Omega^{n} (\mathcal{U (A))} $ be the differential graded algebra of forms defined by $\Omega^0 =\mathcal{ U (A)}$, where ${\mathcal U ( A)}$ is the associative universal enveloping 
algebra  of the ${\Bbb C}$-Lie algebra (\ref{lie}). The smallest subalgebra space $\Omega^{n}_D ({\mathcal U( A)})$ is defined by 
 $$\Omega^{n}_D ({\mathcal U( A)})= \underbrace{\Omega^{1}_D ({\mathcal U (A)})\Omega^{1}_D ({\mathcal U (A)})\cdots\Omega^{1}_D ({\mathcal U (A)})\Omega^{1}_D ({\mathcal U (A)})}_{n},$$
with its elements being finite linear combinations
of monomials of the form $\omega=A_0 d A_1 dA_2 ......dA_n $.  The linear exterior differential $ d:\Omega^{n}_D({\mathcal U( A)})\to \Omega^{n+1}_D({\mathcal U( A)})$
is  defined by 
\be\label{extd}
d (\omega_{1} \otimes_{\mathcal A}\cdots\otimes_{\mathcal A} \omega_{n}) =\sum_{i=1}^{n} (-1)^{i+1} \omega_{1}\otimes_{\mathcal A} \cdots \otimes_{\mathcal A}d\omega_{i} \otimes_{\mathcal A}\cdots \otimes_{\mathcal A}\omega_{n},
\qquad \forall \omega_j \in \Omega^{1}_D ({\mathcal U( A)}),
\ee
and satisfies the basic relations $d^2 =0$ and the extended Leibnitz rule with respect to the product $\otimes_{\mathcal A}$.
The algebra $\Omega^{n}_D ({\mathcal U( A)})$ as defined here is a left $({\mathcal U( A)})$-module but its right structure makes it also a bimodule. The noncommutative universality of the differential $d$ is constructed 
by means of a submodule $\Omega^{1}( \mathcal U( A))= Ker({\mathcal A}\otimes {\mathcal A}\stackrel{m}\rightarrow {\mathcal A})$, such that $m(a\otimes b)=ab$, generated by the Karoubi differential: $da= {\bf 1}\otimes a -a \;\otimes {\bf 1},\quad \forall a\in \mathcal A$. 
Defining now  an anti-derivation operator $i_{\hat X_{k}}$ of degree -1, which is a unique homomorphism $ i_{\hat X_{k}}: \Omega^{1}_{D} (\mathcal U( A))\to E$ on a bimodule $E$, restricted to $\Omega_{D} (\mathcal U(A))$  by  $i_{\hat X_{k}}\triangleright\Omega^{n}_{D} ({\mathcal U( A)})\to \Omega^{n-1}_{D} ({\mathcal U( A)})$, and where
$\hat X_{i}:= D_{i}=i_{\hat X_{i}}\circ d $ is a derivation (we shall be using indistinctly the notation $\hat X_{i}:= D_{i}$ to denote Outer-derivations), on the algebra $\mathcal A$ such that
\be\label{antid2}
i_{\hat X_k}\triangleright\omega(\hat X_1,\dots , \hat X_{n})= (-1)^{k+1} \omega(\hat X_1 , \dots\hat X_{k-1} , \check X_k, \hat X_{k+1}, \dots, \hat X_{n}), 
\ee
 here  $\check X_k $ means removal of $\hat X_k$ from the $n$-form $\omega$. 
In addition we have that for the noncommutative case the Lie-algebra relations
\begin{subequations}
\begin{align}\label{3.4}
&i_{\hat X_i}\circ i_{\hat X_j} +i_{\hat X_j} \circ i_{\hat X_i}=0,\\\label{3.4b}
&L_{\hat X_i}\circ i_{\hat X_j} -i_{\hat X_j} \circ L_{\hat X_i}=i_{[\hat X_i ,\hat X_j]},\\
&L_{\hat X_i}\circ L_{\hat X_j} -L_{\hat  X_j} \circ L_{\hat X_i}=L_{[\hat X_i , \hat X_j]}, 
\end{align}
\end{subequations}
are satisfied 
with $L_{\hat X_i}$ and $i_{ \hat X_i}$ being noncommutative derivation and antiderivation generalizations of the Lie derivative and inner multiplication of forms by vector fields, respectively, and where
\be\label{lie20} 
   L_{\hat X_i}=i_{\hat X_{i}} \circ \;d + d\;\circ i_{\hat X_{i}}
\ee
is a Lie algebra derivation of degree zero of $\Omega_{D} ({\mathcal U (A)}$.

The adjoint action  $\frak {ad}_{  X_{\rho}}=[X_{\rho},\; \boldsymbol\cdot\;] $ is an element of the {\bf Inner} (or Interior) invariant subalgebra $\frak {ad}_{\mathcal A}\in{\mathcal Int}(\mathcal U(\mathcal A))\subset Der(\mathcal U(\mathcal A))$. All other derivations are elements of the {\bf Outer} derivations such that the differential subalgebra $\Omega_{Out}$ is determined by 
\be\label{3.5}
\Omega_{Out}(\mathcal U(\mathcal A))= \big(\omega\in\Omega_{Out}(\mathcal U(\mathcal A)) \; | i_{\hat X_\rho}\triangleright\omega=0 \;\text {and}\; L_{\hat X_{\rho}}\;\omega=0\;\forall\;\ \hat X_{\rho}\in {\mathcal Int({\mathcal A}) }\big).
\ee
Hence the space of all derivations $Der(\mathcal U(\mathcal A))$ is the direct sum  
\begin{equation}\label{ext3}
Der(\mathcal U(\mathcal A))=\mathcal Int(\mathcal U(\mathcal A))\oplus Out (\mathcal U(\mathcal A)),
\ee
where $\mathcal Int(\mathcal U(\mathcal A))$ is a Lie-Ideal and $Out(\mathcal U(\mathcal A))=Der(\mathcal U(\mathcal A))/\mathcal Int(\mathcal U(\mathcal A))$.

Now a derivation on an arbitrary associative algebra is a linear map satisfying the Leibnitz rule
\be\label{leib1}
D_{X_\rho}(X_{\alpha} X_{\beta})=D_{X_\rho}(X_{\alpha})X_{\beta}+X_{\alpha}D_{X_\rho}(X_{\beta}).
\ee
Thus for $D_{X_\rho}, D_{X_\sigma}\in Der(\mathcal U(\mathcal A) $ it follows readily that
\be\label{leib2}
[D_{X_\rho},D_{X_\sigma}]X_{\alpha}X_{\beta}=[\big( D_{X_\rho} \circ D_{X_\sigma}-D_{X_\sigma} \circ D_{X_\rho}\big)X_{\alpha}]X_{\beta} +X_{\alpha}[\big( D_{X_\rho} \circ D_{X_\sigma}-D_{X_\sigma} \circ D_{X_\rho}\big)X_{\beta}],
\ee
so $\big( D_{X_\rho} \circ D_{X_\sigma}-D_{X_\sigma} \circ D_{X_\rho}\big)$ is itself a derivation which is also a $\mathcal{Z(U(A))}$-module and can be written in general as 
\be\label{z-mod1}
\big( D_{X_\rho} \circ D_{X_\sigma}-D_{X_\sigma} \circ D_{X_\rho}\big)=B_{[\rho, \sigma]}^{\tau}\big( \mathcal{Z(U(A))} \big)  D_{X_\tau},
\ee
and 
\be\label{subder4}
 D_{X_{\rho}}X_{\lambda}=\mathcal{N}^{\sigma}_{\rho\lambda}\big(\mathcal {Z(U(A))}\big) X_{\sigma}.
\ee

\subsection{Centrality of the Metric Bimodule}
Now in  order to have well-defined contractions in noncommutative Riemannian geometry we require that the metric two-form $ {\frak g} ={\frak g}_{\mu\nu}(\mathcal U(\mathcal A))\; \omega^\mu \otimes \omega^{\nu}\in \Omega^{2}_D$ be in the center of the Lie algebra.  We thus take the adjoint action of the subalgebra of Inner derivations on it as an isometry
which leaves  it invariant. Hence, by virtue of (\ref{3.5}), the metric 2-form is projected into the subalgebra $\Omega_{Int}(\mathcal U(\mathcal A))$ of {\bf Inner} derivations and,  consequently,
\be\label{linder3}
\frak{ ad}_{ X_\rho}\triangleright\Big( { g}_{\mu\nu}(\mathcal U(\mathcal A)) \; \omega^{\mu} \otimes \omega^{\nu} \Big)=[ X_{\rho},{ g}_{\mu\nu}\big(\mathcal U(\mathcal A)\big)]\big(\omega^\mu \otimes \omega^{\nu}\big) +   { g}_{\mu\nu}(\mathcal U(\mathcal A)) [ X_{\rho},(\omega^\mu \otimes \omega^{\nu})]=0,
 \ee
where the 1-forms $\omega^\mu $ are chosen here as dual to the Inner-derivations $\hat X_{\mu}=\frak {ad}_{X_{\mu}}$.
It then follows by (3.3b) that
\begin{align}
0=L_{\hat X_\rho}\triangleright\omega^{\mu} (\hat X_{\nu}) &=( L_{\hat X_{\rho}}\circ i_{\hat X_{\nu}})\triangleright\omega^\mu =
i_{\hat X_{\nu}}\circ L_{\hat X_{\rho}}\omega^{\mu}+i_{C_{\rho\nu}^{\lambda}({\frak ad}_{X_{\lambda}})}\triangleright\omega^\mu \nonumber\\
&=(L_{\hat X_\rho}\omega^\mu)(\hat X_{\nu}) +C_{\rho\nu}^{\;\;\;\mu} .
\end{align}
The above is clearly satisfied with
\be\label{derdual}
L_{\hat X_{\rho}}\omega^{\mu}=-C^{\mu}_{\rho\sigma} \; \omega^{\sigma}.
\ee

Consequently  (\ref{linder3}) results in
\be\label{add7}
[ X_{\rho}, { g}_{\mu\nu}(\mathcal U(\mathcal A))\; \omega^\mu \otimes \omega^{\nu}]=\Big([ X_{\rho}, { g}_{\mu\nu}\big(\mathcal U(\mathcal A)\big)] -
{ g}_{\sigma\nu}(\mathcal U(\mathcal A))C_{\rho\mu}^{\sigma} - {g}_{\mu\sigma}(\mathcal U(\mathcal A))C_{\rho\nu}^{\sigma}\Big)\omega^\mu \otimes \omega^\nu =0,\\
\ee
where the  $g_{\mu\nu}=g_{\nu\mu}:= g( X_{\mu},  X_{\nu})$ are matrix elements with entries from the translation sector subalgebra (generated by (\ref{trans})), of the  associative unitary enveloping Universal Poincar\'e algebra $\mathcal U(P)$. \\
In order to evaluate the first term on the right of (\ref{add7}) we represent the linear derivation operator $ L_{\hat X_\rho}$ by the adjoint action on the matrix element  
$g( X_{\mu}, X_{\nu})$, {\it i.e.} by $[ X_\rho, g( X_{\mu},  X_{\nu})]$, and evaluate this operator commutator by assuming, as an ansatz, that  $g_{\mu\nu}$ is a polynomial 
series of $X_\rho $ of second order.
Thus, writing 
\be\label{center1}
g( X_{\mu}, X_{\nu})=a_{(\mu\nu)}+ a^{\sigma}_{(\mu\nu)} X_{\sigma} + a^{\sigma\tau}_{(\mu\nu)} X_{\sigma} X_{\tau}
\ee
and, applying (\ref{const1}) we get
\begin{align}\label{cent2}
[ X_{\rho} , g( X_{\mu}, X_{\nu})]=&2ia^{\sigma}_{(\mu\nu)}\big[\xi_{\rho} \big(\eta_{\sigma2} X_{1} -\eta_{\sigma1} X_{2}\big)+\xi_{\sigma} \big(\eta_{\rho1} X_{2} -\eta_{\rho2} X_{1}\big)\big]\nonumber\\
&+2i a^{\sigma\tau}_{(\mu\nu)}\Big(\big[\xi_\rho \big(\eta_{\sigma 2} X_{1} X_{\tau} -\eta_{\sigma 1} X_{2} X_{\tau}\big)+\xi_{\sigma} \big(\eta_{\rho 1} X_2 X_{\tau}-\eta_{\rho 2} X_{1} X_{\tau}\big)\big]+\nonumber\\
&\big[\xi_{\rho} \big(\eta_{\tau2}\ X_\sigma X_1 -\eta_{\tau1} X_\sigma  X_2\big)+\xi_\tau \big(\eta_{\rho1} X_\sigma X_2-\eta_{\rho2}\ X_\sigma X_1\big)\big]\Big).
\end{align}

To complete the calculation of an equation for the metric as the center of the algebra, we observe that the two last terms in the right of (\ref{add7}) are given by
\be\label{sub4}
g([ X_\rho,  X_{\mu}], X_{\nu})=C_{\rho\mu}^{\tau}\;g_{\tau\nu},  \quad  g( X_{\mu}, [ X_\rho, X_{\nu}])=g_{\mu\tau}\;C_{\rho\nu}^{\tau}.
\ee

Therefore, combining (\ref{cent2})  with (\ref{sub4}) and (\ref{sub5}) we obtain, as a condition for centrality, that the coefficients of the universal enveloping algebra should be satisfied by the coefficients
in the following set of equations:
\begin{align}\label{cent3}
&2i a^{\sigma}_{(\mu\nu)}\big[(\xi_0 +\xi_3)\big(\eta_{\sigma2}X_1 -\eta_{\sigma1}X_2\big)+ \xi_\sigma \big(X_1 -X_2 \big)\big]\nonumber +\\
&2i  a^{\sigma\tau}_{(\mu\nu)}\Big(\xi_\sigma \big(X_1 -X_2 \big) X_\tau +\xi_\tau X_\sigma \big(X_1 -X_2 \big) + \nonumber\\
&(\xi_0 +\xi_3)\big[\big(\eta_{\sigma2}X_1 -\eta_{\sigma1}X_2\big) X_\tau + X_\sigma \big(\eta_{\tau 2}X_1 -\eta_{\tau1}X_2\big)\big]\Big) \\
&= 2i\xi_\rho\big[\eta_{\mu2} g_{1\nu} - \eta_{\mu1}g_{2\nu} + \eta_{\nu2}g_{\mu1} - \eta_{\nu1}g_{\mu2}\big].  \nonumber
\end{align}
It readily follows from the above equation that 
the commutator of $X_1$ and $X_2$ with the components $g_{00}, g_{03}, g_{33} $ is zero.
Then the commutation relations are
	\begin{subequations}
      \begin{align}\label{DE1a}
	[{X}_{0},{g}_{\mu\nu}]&=2i\xi_{0}(\eta_{\mu 2}{g}_{1\nu}-\eta_{\mu 1}{g}_{2\nu} + {g}_{\mu 1}\eta_{\nu 2} -{g}_{\mu 2}\eta_{\nu 1}),\\\label{DE2a}
	[{X}_{1},{g}_{\mu\nu}]&=-2i(\xi_{\mu}{g}_{2\nu} + \xi_{\nu}{g}_{\mu 2}),\\\label{DE2b}
	[{X}_{2},{g}_{\mu\nu}]&=2i(\xi_{\mu}{g}_{1\nu} + \xi_{\nu}{g}_{\mu 1}),\\\label{DE2c}
	[{X}_{3},{g}_{\mu\nu}]&=2i\xi_{3}(\eta_{\mu 2}{g}_{1\nu}-\eta_{\mu 1}{g}_{2\nu} + {g}_{\mu 1}\eta_{\nu 2} -{g}_{\mu 2}\eta_{\nu 1}).
	\end{align}
      \end{subequations}\\

\begin{proposition}
	The metric components ${g}_{00}$, ${g}_{03}$ and ${g}_{33}$ depend on polynomials in ${X}_{0}$ and ${X}_{3}$ to the second order at most. 
\end{proposition}
\begin{proof}
	The proof is done by {\it reductio ad absurdum}. So first let us write the algebra needed to prove this for $\hat g_{00}$, i.e.
			\begin{align*}
		[{X}_{2},{g}_{00}]&=2i\xi_{0}({g}_{01}+{g}_{10}),\\
		[{X}_{2},{g}_{01}]&=2i\xi_{0}{g}_{11},\\
		[{X}_{2},{g}_{11}]&=0.
			\end{align*}
	Now let us assume that the metric component has a third degree polynomial term depending on ${X}_{0}$, i.e. $\epsilon {X}_{0}^3$. Due to the algebra and the above commutator relations this would imply that ${g}_{01}$ is of polynomial order two in ${X}_{0}$, this on the other hand implies that $ g_{11}$ is of order one in ${X}_{0}$. Moreover this would also mean that the last commutator relation only holds if and only if $\epsilon=0$. Similar arguments hold for the polynomial degree of ${X}_{3}$ and for the metric components ${g}_{03}$ and ${g}_{33}$.
\end{proof}

Thus from the previous proposition we can write:
\begin{subequations}
      \beqa\label{Comm1a}
	{g}_{00}:={g}_{00}(x_{0},x_{3})=\gamma_{0}+\gamma_{1}{X}_{0}+\gamma_{2}{X}_{3}+\gamma_{3}{X}_{0}{X}_{3}+\gamma_{4}{X}_{0}^{2}+\gamma_{5}{X}_{3}^{2},\\\label{Comm2a}
	{g}_{03}:={g}_{03}(x_{0},x_{3})=\rho_{0}+\rho_{1}{X}_{0}+\rho_{2}{X}_{3}+\rho_{3}\hat{X}_{0}{X}_{3}+\rho_{4}{X}_{0}^{2}+\rho_{5}{X}_{3}^{2},\\\label{Comm3a}
	{g}_{33}:={g}_{03}(x_{0},x_{3})=\kappa_{0}+\kappa_{1}{X}_{0}+\kappa_{2}{X}_{3}+\kappa_{3}{X}_{0}{X}_{3}+\kappa_{4}{X}_{0}^{2}+\kappa_{5}{X}_{3}^{2},
	\eeqa
 \end{subequations}

	where $\gamma_{i}$, $\rho_{i}$ and $\kappa_{i}~ \forall i=0,1,2,3,4,5$ 
are constants. Moreover we can make use of equation (\ref{cas1}) to express $ X_0$ in terms of 
$ X_3$ as
\be\label{cas2}
 X_0 =\frac{1}{\xi_3}\big(\mathcal O_1 + \xi_0  X_{3}),
\ee

which, when substituted into (\ref{Comm1a}-\ref{Comm3a}), results in	
	\begin{subequations}
      \beqa\label{Coeff1a}
	&{\frak g}_{00}:={g}_{00}=a_{0}+a_{1}{X_{3}}+a_{2}{X_{3}}^{2},\\\label{Coeff2a}
	&{\frak g}_{03}:={g}_{03}=b_{0}+b_{1}{X_{3}}+b_{2}{X_{3}}^{2},\\\label{Coeff3a}
	&{\frak g}_{33}:={g}_{33}=c_{0}+c_{1}{X_{3}}+c_{2}{X_{3}}^{2}.
\eeqa
 \end{subequations}

Note that in this way the above  three components of the metric are formally self-adjoint and from here on self-adjoint components of the metric will be denoted by the symbol ${\frak g}_{\mu\nu}$, with their coefficients
 explicitly given by
	\begin{subequations}
      \begin{align}\label{Coeff1b}
a_{0}&=\gamma_{0}+\frac{\gamma_{1}}{\xi_{3}}\mathcal O_1+\frac{\gamma_{4}}{\xi_{3}^{2}}\mathcal O_1^{2},~~a_{1}=\frac{\xi_{0}\gamma_{1}}{\xi_{3}}+\gamma_{2}+\frac{\gamma_{3}}{\xi_{3}}\mathcal O_1+\frac{2\xi_{0}\gamma_{4}}{\xi_{3}^{2}}\mathcal O_1,~~a_{2}=\frac{\xi_{0}\gamma_{3}}{\xi_{3}}+\frac{\xi_{0}^{2}\gamma_{4}}{\xi_{3}^{2}}+\gamma_{5},\\\label{Coeff2b}
b_{0}&=\rho_{0}+\frac{\rho_{1}}{\xi_{3}}\mathcal O_1+\frac{\rho_{4}}{\xi_{3}^{2}}\mathcal O_1^{2},~~b_{1}=\frac{\xi_{0}\rho_{1}}{\xi_{3}}+\rho_{2}+\frac{\rho_{3}}{\xi_{3}}{\mathcal O}_1+\frac{2\xi_{0}\rho_{4}}{\xi^2_{3}} {\mathcal O}_1,~~b_{2}=\frac{\xi_{0}\rho_{3}}{\xi_{3}}+\frac{\xi_{0}^{2}\rho_{4}}{\xi_{3}^{2}}+\rho_{5},\\\label{Coeff3b}
c_{0}&=\kappa_{0}+\frac{\kappa_{1}}{\xi_{3}}\mathcal O_1+\frac{\kappa_{4}}{\xi_{3}^{2}}\mathcal O_1^{2},~~c_{1}=\frac{\xi_{0}\kappa_{1}}{\xi_{3}}+\kappa_{2}+\frac{\kappa_{3}}{\xi_{3}}\mathcal O_1+\frac{2\xi_{0}\kappa_{4}} {\xi_{3}^{2}}{\mathcal O}_1,~~c_{2}=\frac{\xi_{0}\kappa_{3}}{\xi_{3}}+\frac{\xi_{0}^{2}\kappa_{4}}{\xi_{3}^{2}}+\kappa_{5}.
	\end{align}
 \end{subequations}

The remaining metric components are then derived from the commutation relations (\ref{DE1a}-\ref{DE2c}) and the last two sets of relations to yield  
	\begin{subequations}
      \begin{align}\label{Coeff1c}
	{\frak g}_{01}&=\alpha{X_{1}}+\beta({X_{1}}{X_{3}}+{X_{3}}{X_{1}}),\\\label{Coeff2c}
{\frak g}_{02}&=\alpha{X_{2}}+\beta({X_{2}}{X_{3}}+{X_{3}}{X_{2}}),\\\label{Coeff3c}
	{\frak g}_{13}&=\alpha^{\prime}{X_{1}}+\beta^{\prime}({X_{1}}{X_{3}}+{X_{3}}{X_{1}}),\\\label{Coeff4c}
	 {\frak g}_{23}&=\alpha^{\prime}{X_{2}}+\beta^{\prime}({X_{2}}{X_{3}}+{X_{3}}{X_{2}}) ,\\\label{Coeff5c}
	 {\frak g}_{11}&=\delta{X_{1}}^{2},\\\label{Coeff6c}
	{\frak g}_{22}&=\delta{X_{2}}^{2},\\\label{Coeff7c}
{\frak g}_{12}&=\frac{\delta}{2} ({X_{1}}{X_{2}}+{X_{2}}{X_{1}})=\delta {X_{1}}{X_{2}},
     \end{align}
     \end{subequations}

where
\be\label{adj}
g_{ji}= g^\dag_{ij}, 	\qquad {\frak g}_{ij}:=\frac{1}{2}(g_{ij}+g^\dag _{ij}) \;\;\;{\text for}\;\;\;\;i\leq j 
\ee
and 
\begin{equation}\label{consrel1}
	\alpha=\frac{\xi_{3}}{2\xi_{0}}a_{1}, \qquad \beta =\frac{\xi_{3}}{2\xi_{0}}a_{2}, \qquad \alpha^{\prime}=\frac{c_{1}}{2}, \qquad \beta^{\prime}=\frac{c_{2}}{2} \qquad \delta=\Big(\frac{\xi_{3}}{\xi_{0}}\Big)^{2}a_{2}\\
      \end{equation} 

	Next note that, due to the commutations relations, we will have relations between the coefficients $a_{i}$, $b_{i}$ and $c_{i}$. Indeed from (\ref{DE2a}) and (\ref{DE2b}), we have
	\begin{align*}
	[{X_{1}},{\frak g}_{03}]&=-2i(\xi_{0}{\frak g}_{23}+\xi_{3}{\frak g}_{02})\\
	[{X_{2}},{\frak g}_{03}]&=-2i(\xi_{0}{\frak g}_{13}+\xi_{3}{\frak g}_{01}),
	\end{align*}

	which, together with (\ref{Coeff2a}), (\ref{Coeff1c}-\ref{Coeff4c}) and (\ref{consrel1}), results in
	\begin{align}\label{b1}
     & b_1 = \big[\big(\frac{\xi_3}{2\xi_0}\big)a_1 + \big(\frac{\xi_0}{2\xi_3}\big)c_1\big]\\\label{b2}
     & b_2 = \big[\big(\frac{\xi_3}{2\xi_0}\big)a_2 + \big(\frac{\xi_0}{2\xi_3}\big)c_2\big].
	\end{align}

In a similar fashion, making use of (\ref{DE2a}) for $\mu=1$, $\nu=3$ and (\ref{Coeff7c}), we obtain
     \be\label{consrel2}
      c_2 = a_2 \big(\frac{\xi_3}{\xi_0}\big)^2=\delta . 
      \ee
     Replacing (\ref{consrel2}) in (\ref{b2}) yields
     \be\label{consrel3}
      b_2= \big(\frac{\xi_3}{\xi_0}\big)a_2=2\beta,\\
      \ee
with which we obtain a relation between the coefficients that are independent of $\mathcal{O}_{1}$.     
      
In this way the full  expression for the metric is given by
\begin{equation*}\label{qmet}
{{\frak g}_{\mu \nu }} = \left( {\begin{array}{*{20}{c}}
{{a_0} + {a_1}Z + {a_2}{Z^2}}&{\alpha X + \beta (XZ + ZX)}&{\alpha Y + \beta (YZ + ZY)}&{{b_0} + {b_1}Z + {b_2}{Z^2}}\\
{\alpha X + \beta (XZ + ZX)}&{\delta {X^2}}&{\delta XY}&{\alpha^{\prime} X + \beta '(XZ + ZX)}\\
{\alpha Y + \beta (YZ + ZY)}&{\delta XY}&{\delta {Y^2}}&{\alpha 'Y + \beta '(YZ + ZY)}\\
{{b_0} + {b_1}Z + {b_2}{Z^2}}&{\alpha 'X + \beta '(XZ + ZX)}&{\alpha 'Y + \beta '(YZ + ZY)}&{{c_0} + {c_1}Z + {c_2}{Z^2}}
\end{array}} \right),
\end{equation*}
where for simplicity we have put $X_{1}=X$, $X_{2}=Y$ and $X_{3}=Z$.

In order to analyze whether the quantum metric $g_{\mu\nu}$ is non-singular, 
let us now consider the determinant of the matrix ${\frak g}=[ {\frak g}_{\mu\nu}]_{\mu,\nu=0,1,2,3}$ with non-commuting entries,  defined as the sum \cite{kull}:
	\begin{align}\label{qdet}
	\det{\frak g}=\sum_{\sigma \in \frak{A}_{4}} {\frak g}_{\sigma(0)0}\cdot{\frak g}_{\sigma(1)1}\cdot{\frak g}_{\sigma(2)2}\cdot{\frak g}_{\sigma(3)3}\cdot  sign\sigma,    
	\end{align}
	where $\sigma$ are the permutations of the antisymmetrizer $\frak{A}_{4}$. Note that this preserves the appropriate ordering of the noncommutative entries of the matrix and is equivalent
     to summing sequentially the product of the elements of the first columns by their corresponding cofactors.

Explicitly this  determinant is given, factoring similar terms, by 
\begin{align}\label{quantumdeter}. 
\begin{split}
\det{\frak g}&={\frak g}_{00}[({\frak g}_{11}{\frak g}_{22}-{\frak g}_{21}{\frak g}_{12}){\frak g}_{33}+({\frak g}_{31}{\frak g}_{12}-{\frak g}_{11}{\frak g}_{32}){\frak g}_{23}+({\frak g}_{21}{\frak g}_{32}-{\frak g}_{31}
{\frak g}_{22}){\frak g}_{13}] \\
&+{\frak g}_{30}[({\frak g}_{11}{\frak g}_{02}-{\frak g}_{01}{\frak g}_{12}){\frak g}_{23}+({\frak g}_{01}{\frak g}_{22}-{\frak g}_{21}{\frak g}_{02}){\frak g}_{13}-({\frak g}_{11}{\frak g}_{22}-{\frak g}_{21}{\frak g}_{12})\
{\frak g}_{03}]  \\
&+{\frak g}_{10}[({\frak g}_{01}{\frak g}_{32}-{\frak g}_{31}{\frak g}_{02}){\frak g}_{23}-({\frak g}_{01}{\frak g}_{22}-{\frak g}_{21}{\frak g}_{02}){\frak g}_{33}-({\frak g}_{21}{\frak g}_{32}-{\frak g}_{31}{\frak g}_{22})
{\frak g}_{03}]  \\
 &-{\frak g}_{20}[({\frak g}_{11}{\frak g}_{02}-{\frak g}_{01}{\frak g}_{12}){\frak g}_{33}+({\frak g}_{01}{\frak g}_{32}-{\frak g}_{31}{\frak g}_{02}){\frak g}_{13}+({\frak g}_{31}{\frak g}_{12}-{\frak g}_{11}{\frak g}_{32})
{\frak g}_{03}].
\end{split} 
\end{align}
     	
Our essential objective is to show that this determinant is non-singular. Towards this end we make use of the equations  (\ref{Coeff1a}-\ref{Coeff3a}), (\ref{Coeff1c}-\ref{Coeff7c}) and the commutation relations to get
    
    \begin{subequations}
    \begin{align}\label{ME1}
	J_{1}&:={\frak g}_{11}{\frak g}_{22}-{\frak g}_{21}{\frak g}_{12}=0,\\\label{ME2}
	J_{2}&:={\frak g}_{31}{\frak g}_{12}-{\frak g}_{11}{\frak g}_{32}=-\xi_{3}\xi_{0}^{-1}J_{4},\\\label{ME3}
	J_{3}&:={\frak g}_{21}{\frak g}_{32}-{\frak g}_{31}{\frak g}_{22}=-\xi_{0}\xi_{3}^{-1}J_{5},\\\label{ME4}
	J_{4}&:={\frak g}_{11}{\frak g}_{02}-{\frak g}_{01}{\frak g}_{12}=2i\xi_{3}\beta\delta X_{1}(4X_{1}^{2}-3\mathcal{O}_{2}),\\\label{ME5}
	J_{5}&:={\frak g}_{01}{\frak g}_{22}-{\frak g}_{21}{\frak g}_{02}=2i\xi_{3}\beta\delta X_{2}(\mathcal{O}_{2}-4X_{1}^{2}),\\\label{ME6}
	J_{6}&:={\frak g}_{01}{\frak g}_{32}-{\frak g}_{31}{\frak g}_{02}=2i\xi_{3}\beta(\alpha^{'}-\xi_{0}^{-1}\xi_{3}\alpha)(\mathcal{O}_{2}-X_{1}^{2}),
	\end{align}
	\end{subequations}
	
Using the equations (\ref{ME1}-\ref{ME6}) and applying the commutation repeatedly (see details in the Appendix A) we have shown that the determinant is nonsingular and it is given by 
\begin{equation}\label{app}
\det \frak{{g}}_{\mu\nu}=K_{0}+X_{1}^{2}K_{1}+X_{1}X_{2}K_{2},
\end{equation}	
with $K_i=K_i({\mathcal Z(\mathcal A)},X_{3})$ for $i=0,1,2$.\\

\subsection{Inverse of Quantum Metric}

Next in order to derive the inverse of the above metric matrix we make use of the concept of quasi-determinants discussed in \cite{GG}. Thus we express the matrix metric in (\ref{qmet}) in terms of the following $2\times 2$ blocks
\be\label{qmet2}
{G} = \left( {\begin{array}{*{20}{cc|cc}}
{{{\frak g}_{00}}}&{{{\frak g}_{01}}}&{{{\frak g}_{02}}}&{{{\frak g}_{03}}}\\
{{{\frak g}_{10}}}&{{{\frak g}_{11}}}&{{{\frak g}_{12}}}&{{{\frak g}_{13}}}\\
\hline
{{{\frak g}_{20}}}&{{{\frak g}_{21}}}&{{{\frak g}_{22}}}&{{{\frak g}_{23}}}\\
{{{\frak g}_{30}}}&{{{\frak g}_{31}}}&{{{\frak g}_{32}}}&{{{\frak g}_{33}}}
\end{array}} \right) = \left( {\begin{array}{*{20}{c|c}}
{{G_{11}}}&{{G_{12}}}\\
\hline
{{G_{21}}}&{{G_{22}}}
\end{array}} \right).
\ee

If in addition we let 
\be\label{invqm}
{Y}=\left( {\begin{array}{*{20}{c|c}}
{{Y_{11}}}&{{Y_{12}}}\\
\hline
{{Y_{21}}}&{{Y_{22}}}
\end{array}} \right),
\ee
where 
\begin{align}\label{invqm2} 
Y_{11}=&(G_{11} - G_{12}G_{22}^{-1} G_{21})^{-1}\nonumber\\
Y_{12}=& - G_{11}^{-1} G_{12}Y_{22}\\
Y_{21}=&- G_{22}^{-1} G_{21}Y_{11}\nonumber\\
Y_{22}=&(G_{22} - G_{21}G_{11}^{-1} G_{12})^{-1},\nonumber
\end{align}
it can be readily shown that these also $2\times 2$ block matrices are such that (\ref{invqm}) is the right-inverse matrix $Y= G^{-1}$, {\it i.e.}
${GY} =\bigl(\begin{smallmatrix}
       I_2 & 0 \\ 0 & I_2
       \end{smallmatrix}\bigr)$. Substituting the inverses of the above metric blocks into (\ref{invqm2}) one obtains after a rather lengthy calculation an explicit expression for the right inverse $Y$ of the metric.
       
Let $\frak{y}^{\mu\nu}$ be the components of the matrix $Y$ defined above, our goal is to find these components explicitly. For this purpose, we start by finding the components of the matrices $G_{11} - G_{12}G_{22}^{-1} G_{21}$ and \\
$G_{22} - G_{21}G_{11}^{-1} G_{12}$ in the following way:
\begin{align*}
G_{11} - G_{12}G_{22}^{-1} G_{21}&= \left( \begin{array}{cc}
\frak{g}_{00} & \frak{g}_{01}\\
\frak{g}_{10} & \frak{g}_{11}
\end{array}\right)-\left( \begin{array}{cc}
\frak{g}_{02} & \frak{g}_{03}\\
\frak{g}_{12} & \frak{g}_{13}
\end{array}\right)\left( \begin{array}{cc}
\frak{s}_{22} & \frak{s}_{23}\\
\frak{s}_{32} & \frak{s}_{33}
\end{array}\right)\left( \begin{array}{cc}
\frak{g}_{20} & \frak{g}_{21}\\
\frak{g}_{30} & \frak{g}_{31}
\end{array}\right)  =\left( \begin{array}{cc}
\frak{l}_{00} & \frak{l}_{01}\\
\frak{l}_{10} & \frak{l}_{11}
\end{array}\right),\\
G_{22} - G_{21}G_{11}^{-1} G_{12}&= \left( \begin{array}{cc}
\frak{g}_{22} & \frak{g}_{23}\\
\frak{g}_{32} & \frak{g}_{33}
\end{array}\right)-\left( \begin{array}{cc}
\frak{g}_{20} & \frak{g}_{21}\\
\frak{g}_{30} & \frak{g}_{31}
\end{array}\right)\left( \begin{array}{cc}
\frak{s}_{00} & \frak{s}_{01}\\
\frak{s}_{10} & \frak{s}_{11}
\end{array}\right)\left( \begin{array}{cc}
\frak{g}_{02} & \frak{g}_{03}\\
\frak{g}_{12} & \frak{g}_{13}
\end{array}\right)  =\left( \begin{array}{cc}
\frak{l}_{22} & \frak{l}_{23}\\
\frak{l}_{32} & \frak{l}_{33}
\end{array}\right),
\end{align*}
where $\frak{s}_{\mu\nu}$ are the components of the inverse matrices $G_{11}$, $G_{22}$ obtained by applying the Gel'fand algorithm and are given by
\begin{align}\nonumber
\frak{s}_{00}&=(\frak{g}_{00}-\frak{g}_{01}\frak{g}_{11}^{-1}\frak{g}_{10})^{-1},    &\qquad & \frak{s}_{22}=(\frak{g}_{22}-\frak{g}_{23}\frak{g}_{33}^{-1}\frak{g}_{32})^{-1},\\\label{inverseGij}
\frak{s}_{01}&=-\frak{g}_{00}^{-1}\frak{g}_{01}\frak{s}_{11},        & \qquad &  \frak{s}_{23}=-\frak{g}_{22}^{-1}\frak{g}_{23}\frak{s}_{33},\\\nonumber
\frak{s}_{10}&=-\frak{g}_{11}^{-1}\frak{g}_{10}\frak{s}_{00}, &\qquad &  \frak{s}_{32}=-\frak{g}_{33}^{-1}\frak{g}_{32}\frak{s}_{22},\\\nonumber
\frak{s}_{11}&=(\frak{g}_{11}-\frak{g}_{10}\frak{g}_{00}^{-1}\frak{g}_{01})^{-1}, &\qquad &  \frak{s}_{33}=(\frak{g}_{33}-\frak{g}_{32}\frak{g}_{22}^{-1}\frak{g}_{23})^{-1};
\end{align}
while the explicit expressions for $\frak{l}_{\mu\nu}$ are: 
\begin{align}\nonumber
\frak{l}_{\kappa\lambda}&=\frak{g}_{\kappa\lambda}-(\frak{g}_{\kappa2}-\frak{g}_{\kappa3}\frak{g}_{33}^{-1}\frak{g}_{32})\frak{s}_{22}\frak{g}_{2\lambda}-(\frak{g}_{\kappa3}-\frak{g}_{\kappa2}\frak{g}_{22}^{-1}\frak{g}_{23})\frak{s}_{33}\frak{g}_{3\lambda},\\\label{lcomponents}
\frak{l}_{\tau\sigma}&=\frak{g}_{\tau\sigma}-(\frak{g}_{\tau0}-\frak{g}_{\tau1}\frak{g}_{11}^{-1}\frak{g}_{10})\frak{s}_{00}\frak{g}_{0\sigma}-(\frak{g}_{\tau1}-\frak{g}_{\tau0}\frak{g}_{00}^{-1}\frak{g}_{01})\frak{s}_{11}\frak{g}_{1\sigma},
\end{align}
for $\kappa, \lambda=0,1$ and $\tau, \sigma=2,3$.

Finally, we can calculate the inverse of the blocks given by the equations (\ref{invqm2}) by applying again the method of Gel'fand to obtain the blocks $Y_{11}$ and $Y_{22}$, with which we obtain the rest by means of the matrix multiplication. Therefore the components of the right inverse read:
\be\label{qminversecomp}
{Y} = \left( {\begin{array}{*{20}{cc|cc}}
{{{\frak y}^{00}}}&{{{\frak y}^{01}}}&{{{\frak y}^{02}}}&{{{\frak y}^{03}}}\\
{{{\frak y}^{10}}}&{{{\frak y}^{11}}}&{{{\frak y}^{12}}}&{{{\frak y}^{13}}}\\
\hline
{{{\frak y}^{20}}}&{{{\frak y}^{21}}}&{{{\frak y}^{22}}}&{{{\frak y}^{23}}}\\
{{{\frak y}^{30}}}&{{{\frak y}^{31}}}&{{{\frak y}^{32}}}&{{{\frak y}^{33}}}
\end{array}} \right) = \left( {\begin{array}{*{20}{cc|cc}}
\frak{y}^{00}& -\frak{l}_{00}^{-1}\frak{l}_{01}\frak{y}^{11}&\frak{N}_{1}\frak{y}^{22}&\frak{N}_{2}\frak{y}^{33}\\
-\frak{l}_{11}^{-1}\frak{l}_{10}\frak{y}^{00}& \frak{y}^{11} &\frak{N}_{3}\frak{y}^{22}&\frak{N}_{4}\frak{y}^{33}\\
\hline
\frak{R}_{1}\frak{y}^{00}&\frak{R}_{2}\frak{y}^{11}&{{{\frak y}^{22}}}&-\frak{l}_{22}^{-1}\frak{l}_{23}\frak{y}^{33}\\
\frak{R}_{3}\frak{y}^{00}&\frak{R}_{4}\frak{y}^{11}&-\frak{l}_{33}^{-1}\frak{l}_{32}\frak{y}^{22}&{{{\frak y}^{33}}}
\end{array}} \right),
\ee
where we have defined the quantum operators $\frak{N}_{i}$ as:
\begin{align}\nonumber
\frak{N}_{1}&=-(\frak{s}_{00}\frak{g}_{02}-\frak{g}_{00}^{-1}\frak{g}_{01}\frak{s}_{11}\frak{g}_{12})+(\frak{s}_{00}\frak{g}_{03}-\frak{g}_{00}^{-1}\frak{g}_{01}\frak{s}_{11}\frak{g}_{13})\frak{l}_{33}^{-1}\frak{l}_{32}, \\\label{Noperators}
\frak{N}_{2}&=-(\frak{s}_{00}\frak{g}_{03}-\frak{g}_{00}^{-1}\frak{g}_{01}\frak{s}_{11}\frak{g}_{13})+(\frak{s}_{00}\frak{g}_{02}-\frak{g}_{00}^{-1}\frak{g}_{01}\frak{s}_{11}\frak{g}_{12})\frak{l}_{22}^{-1}\frak{l}_{23}, \\\nonumber
\frak{N}_{3}&=-(\frak{s}_{11}\frak{g}_{12}-\frak{g}_{11}^{-1}\frak{g}_{10}\frak{s}_{00}\frak{g}_{02})+(\frak{s}_{11}\frak{g}_{13}-\frak{g}_{11}^{-1}\frak{g}_{10}\frak{s}_{00}\frak{g}_{03})\frak{l}_{33}^{-1}\frak{l}_{32}, \\\nonumber
\frak{N}_{4}&=-(\frak{s}_{11}\frak{g}_{13}-\frak{g}_{11}^{-1}\frak{g}_{10}\frak{s}_{00}\frak{g}_{03})+(\frak{s}_{11}\frak{g}_{12}-\frak{g}_{11}^{-1}\frak{g}_{10}\frak{s}_{00}\frak{g}_{02})\frak{l}_{22}^{-1}\frak{l}_{23}, 
\end{align}
and $\frak{R}_{i}$ as
\begin{align}\nonumber
\frak{R}_{1}&=-(\frak{s}_{22}\frak{g}_{20}-\frak{g}_{22}^{-1}\frak{g}_{23}\frak{s}_{33}\frak{g}_{30})+(\frak{s}_{22}\frak{g}_{21}-\frak{g}_{22}^{-1}\frak{g}_{23}\frak{s}_{33}\frak{g}_{31})\frak{l}_{11}^{-1}\frak{l}_{10}, \\\label{Roperators}
\frak{R}_{2}&=-(\frak{s}_{22}\frak{g}_{21}-\frak{g}_{22}^{-1}\frak{g}_{23}\frak{s}_{33}\frak{g}_{31})+(\frak{s}_{22}\frak{g}_{20}-\frak{g}_{22}^{-1}\frak{g}_{23}\frak{s}_{33}\frak{g}_{30})\frak{l}_{00}^{-1}\frak{l}_{01}, \\\nonumber
\frak{R}_{3}&=-(\frak{s}_{33}\frak{g}_{30}-\frak{g}_{33}^{-1}\frak{g}_{32}\frak{s}_{22}\frak{g}_{20})+(\frak{s}_{33}\frak{g}_{31}-\frak{g}_{33}^{-1}\frak{g}_{32}\frak{s}_{22}\frak{g}_{21})\frak{l}_{11}^{-1}\frak{l}_{10}, \\\nonumber
\frak{R}_{4}&=-(\frak{s}_{33}\frak{g}_{31}-\frak{g}_{33}^{-1}\frak{g}_{32}\frak{s}_{22}\frak{g}_{21})+(\frak{s}_{33}\frak{g}_{30}-\frak{g}_{33}^{-1}\frak{g}_{32}\frak{s}_{22}\frak{g}_{20})\frak{l}_{00}^{-1}\frak{l}_{01}, 
\end{align}
and where the diagonal independent components of the inverse (\ref{qminversecomp}) are given, as stated earlier, by:
\begin{align}\nonumber
\frak{y}^{00}&=(\frak{l}_{00}-\frak{l}_{01}\frak{l}_{11}^{-1}\frak{l}_{10})^{-1}, &  & \frak{y}^{22}=(\frak{l}_{22}-\frak{l}_{23}\frak{l}_{33}^{-1}\frak{l}_{32})^{-1},\\\label{indepQMcomp}
\frak{y}^{11}&=(\frak{l}_{11}-\frak{l}_{10}\frak{l}_{00}^{-1}\frak{l}_{01})^{-1},
&  & \frak{y}^{33}=(\frak{l}_{33}-\frak{l}_{32}\frak{l}_{22}^{-1}\frak{l}_{23})^{-1},
\end{align}
Note that it is easy to show by a fairly straightforward calculation, that $\frak{g}_{\mu\rho}{(\frak y^{T})}^{\rho\nu}=\delta_{\mu}^{\nu}$ and by (\ref{qminversecomp}) it is clearly seen that the inverse of the metric is not symmetric in its entries while the metric $G$ itself is. Thus to derive the left action of $Y$  on $G$, {\it i.e.}
${Y^{T} G} =\bigl(\begin{smallmatrix}
       I_2 & 0 \\ 0 & I_2
       \end{smallmatrix}\bigr)$, 

we need to take the transpose of the matrix blocks in $Y$ and transpose each of these blocks. Explicitly:
 \be\label{invqm7}
{Y^T}=\left( {\begin{array}{*{20}{c|c}}
{{Y^{T}_{11}}}&{{Y^{T}_{21}}}\\
\hline
{{Y^{T}_{12}}}&{{Y^{T}_{22}}}
\end{array}} \right) = \left( {\begin{array}{*{20}{cc|cc}}
{{{(\frak y^{T})}^{00}}}&{{(\frak y^{T})}^{10}}&{{(\frak y^{T}})^{20}}&{{{(\frak y^{T}})}^{30}}\\
{{{(\frak y^{T})^{01}}}}&{{{(\frak y^{T})}^{11}}}&{{{(\frak y^{T})}^{21}}}&{{{(\frak y^{T})}^{31}}}\\
\hline
{{{(\frak y^{T})}^{02}}}&{{{(\frak y^{T})}^{12}}}&{{{(\frak y^{T})}^{22}}}&{{{(\frak y^{T}})^{32}}}\\
{{{(\frak y^{T}) }^{03}}}&{{{(\frak y^{T})}^{13}}}&{{{(\frak y^{T})}^{23}}}&{{{(\frak y^{T})}^{33}}}
\end{array}} \right),
\ee
and where in $(\frak y^{T})^{\rho\nu}$ the $(T)$ upper index means the transpose of the entries in the components of the inverse in (\ref{qminversecomp}) while the $(\rho\nu)$-Leibnitz indices refer to the column and row position, respectively, of the inverse in that matrix.\\
As an illustration of how the inverse of the elements of the metric matrix act on the Lorentzian indices of a tensor, consider a covariant two-tensor $R_{\mu\nu}$ which has been defined. To construct a mixed tensor from it, start with the intrinsic equation $R_{\mu\nu}=:g_{\mu\sigma}R^{\sigma}_{\;\nu }$. Using now the inverse metric components displayed in (\ref{invqm7}) and acting with these from the left on both sides of the previous definition we get
\begin{equation}\label{inv6}
 ( \frak  y^{T})^{\lambda\mu}R_{\mu \nu} =( \frak  y^{T})^{\lambda\mu}g_{\mu\sigma}R^{\sigma}_{\;\nu }=\delta ^{\lambda}_{\sigma}R^{\sigma}_{\;\nu }=R^{\lambda}_{\;\nu }.
\end{equation}
Note, by analogy, that 
\be\label{inv7}
R_{\mu\lambda}\;\frak  y^{\lambda\nu}=R_{\mu}^{\:\;\sigma} g_{\sigma\lambda}\;\frak  y^{\lambda\nu}=R_{\mu}^{\;\:\sigma}\delta_{\sigma}^{\nu}=R_{\mu}^{\;\:\nu}.
\ee

 In the next section we shall implement the above results in the consideration of the Pseudo-Riemannian tensor associated with our algebra.

\section{ Noncommutative Riemannian Geometry from the Metric Central Bimodule}

 From the theory of Derivations and Noncommutative Differential Calculus \cite{connes1}, \cite{dubois1},\cite{dubois2},\cite{landi},\cite{jacobson} we have that the Leibnitz rule action of the linear mapping 
$(\boldsymbol\nabla ,\hat X_{i})\mapsto \nabla_{ X_{i}}$ is a linear connection endomorphism (from here on, without risk of confusion and in order to simplify our notation, we shall denote the  associative enveloping universal algebra with unit $(\mathcal U(\mathcal A))$ simply by $\mathcal A$).  In particular, making use of (\ref{z-mod1}) \be\label{covder1}
\nabla_{ X_{i}} (X_{\alpha}\; \psi) = D_i X_{\alpha} \;\psi + X_{\alpha}\; \nabla_{ X_{i}}(\psi), \qquad \forall \; X_{\alpha} \in{\mathcal A} , 
\ee
where $\psi$ is a left $\mathcal A$-module,  we can equally consider
\be\label{covder2}
\nabla_{ X_{\rho}} (X_{\alpha}\; \psi) = \hat X_{\rho} (X_{\alpha})\;\psi + X_{\alpha}\; \nabla_{ X_{\rho}}(\psi), \qquad \forall \; X_{\alpha} \in{\mathcal A}.  
\ee
Now, when acting with this covariant derivation on the metric two-form ${\frak g}:= {\frak g}_{\mu\nu}(\mathcal A)\; \omega^\mu \otimes \omega^{\nu} \in \Omega^{2}(\mathcal A))$
the latter is projected into the subspace of {\bf Outer}-derivations so
\be\label{linder1}
\nabla_{ X_{\rho}} \big({\frak g}_{\mu\nu}(\mathcal A)\omega^\mu \otimes\omega^\nu\big)= D_{ X_{\rho}}\big({\frak g}_{\mu\nu}( \mathcal A)\big)\omega^\mu \otimes\omega^\nu + 
{\frak g}_{\mu\nu}(\mathcal A)\nabla_{X_{\rho}}(\omega^\mu \otimes\omega^\nu),
\ee
where the differential subalgebra $\Omega_{Out}$ is determined by  (\ref{3.5}). Thus
 $D_{ X_\rho}\in Out(\mathcal A)= Der(\mathcal A)/{\mathcal Int({\mathcal A}) }$ is the outer-derivation associated with the Lie algebra basis $X_\rho$, as defined in Sec.(3.1).

Note also that because we are using the directional covariant derivation we do not require a flip factor in (\ref{linder1}) to preserve operator ordering. Hence, in particular, we have from (\ref{covder2})
  \begin{align}\label{linder2}
\nabla_{X_{\rho}}( X_{\nu}\;\omega^\mu)=D_{ X_{\rho}} ( X_{\nu})\;\omega^\mu  + X_{\nu} \nabla_{ X_{\rho}}(\omega^\mu). 
  \end{align}
Now, by linearity, we can set \;\;\;
\be\label{cons3d}
\nabla_{X_{\rho}}\;\omega^{\mu}=-\Gamma_{\rho\sigma}^{\mu}\; \omega^\sigma,
\ee
where, in order to preserve stability in the
$\Omega_D$ cochains we require that $\Gamma_{\nu\sigma}^{\mu}$ be in the center $\mathcal (Z(\mathcal A))$ of the algebra, as it can be easily shown that $Der({\mathcal A})$ maps $\mathcal Z(\mathcal A)$ into itself \cite{Landi2}, so $ D_{X_{\nu}}( \mathcal Z(\mathcal A))\in \mathcal Z(\mathcal A)$ and $Der{\mathcal Z(\mathcal A)}\subset{\mathcal Out(\mathcal A)}$.

To further specify these connection symbols and relate to the Levi-Civita connection of Pseudo-Riemannian Geometry, we next consider the following two properties which will  result in a noncommutative generalization  of  :\\
 1){\bf Metricity}. The covariant derivative of the metric tensor ${\frak g}$ is required to vanish, so that 
 \be\label{met}
  \nabla_{ X_{\rho}} \big( {\frak g}_{\mu,\nu} \; \omega^{\mu}\otimes\omega^{\nu}\big)=  0.
\ee

2) {\bf Zero torsion}.
Given a connection $\nabla_{ X_{\rho}} $ on $\Omega_{Der} (\mathcal A)$, torsion is defined as the bimodule homomorphism $T: \Omega^{1}_{Der} (\mathcal A)\to\Omega^{2}_{Der} (\mathcal A) $
 by setting \cite{D-V3}:
\be\label{tor1}
T(\omega)(\hat X_{\alpha} , \hat X_{\beta})= d\omega(\hat X_{\alpha}, \hat X_{\beta})-\nabla_{X_{\alpha}}( \omega)(\hat X_{\beta})+\nabla_{X_{\beta}}( \omega)(\hat X_{\alpha}),\quad \hat X_{\alpha} , \hat X_{\beta}\in {Der} (\mathcal A)
\ee

In order to further solve (\ref{tor1}) for zero torsion, recall first that the differential $d$ of $\Omega^{n}_D$ in (\ref{extd}) can be expressed equivalently as 
\begin{align}\label{tora}
(d\omega)(\hat X_1, \dots , \hat X_{n+1})=& \sum^{n+1}_{k=1} (-1)^{k+1} \hat X_{k}\; \omega(\hat X_{1},\dots, \check X_k, \dots, \hat X_{n+1})\nonumber\\
&+\sum_{k<l}^{n+1} (-1)^{k+l}\; \omega([\hat X_k , \hat X_l] ,\hat X_1\dots, \check X_k,\dots \check X_l,\dots \hat X_{n+1}),
\end{align}
from where, in particular,  it follows that 
\be\label{tora2}
(d\omega)(\hat X_\alpha , \hat X_{\beta})=\hat X_{\alpha} \big(\omega(\hat X_{\beta})\big)-\hat X_\beta \big(\omega(\hat X_\alpha)\big)-\omega([\hat X_\alpha,  \hat X_{\beta}]).
\ee
Writing now the general 1-form as $\omega= \sum_{\sigma}X_{\sigma} \omega^{\sigma}$ and making use of the duality $\omega^{\sigma}(\hat X_\gamma)= \delta^{\sigma}_{\gamma}$ and (\ref{subder4}), yields
\be\label{tora3}
\omega(\hat X_\beta)=i_{\hat X_\beta}\triangleright(\sum_{\sigma}X_{\sigma} \omega^{\sigma})=X_{\beta},
\ee
and 
\be\label{tora4}
\hat X_\alpha \big(\omega(\hat X_\beta) \big)=D_{X_{\alpha}}X_\beta= \mathcal{N}^{\sigma}_{\alpha\beta}\big(\mathcal {Z(U(A))}\big) X_{\sigma}.
\ee
Moreover, recalling now (\ref{z-mod1}), we get
\be\label{tora5}
\omega([\hat X_\alpha, \hat X_\beta])=\omega\big (B^{\lambda}_{[\alpha,\beta]}\hat X_{\lambda})=B^{\sigma}_{[\alpha,\beta]}X_\sigma .
\ee
As a next step consider the covariant derivative terms in (\ref{tor1}) by making use of (\ref{cons3d}). By a similar procedure as before we get
\be\label{tora6}
\nabla_{X_\alpha}(\omega)(\hat X_\beta)=[D_{X_{\alpha}}X_{\beta}-X_{\sigma}\Gamma^{\sigma}_{\alpha\beta}].
\ee
Hence replacing the terms  (\ref{tora4})and (\ref{tora5}) we can calculate the exterior derivative in (\ref{tora2}) and substituting these together with (\ref{tora6}) in (\ref{tor1}) 
results in 
\be\label{tor3}
(\Gamma^\sigma_{\mu\nu} -\Gamma^\sigma_{\nu\mu})=  B_{[\mu\nu]}^{\sigma}\big(\mathcal {Z(U(A))}\big) .\\
\ee
Consequently zero torsion implies that the antisymmetric part of the symbols $\Gamma_{\nu\sigma}^{\mu}(\mathcal Z(\mathcal A))$ are determined by the $(B_{[\alpha, \beta]}^{\sigma})$ in (\ref{z-mod1}).

All this, of course, reflects the fact that while in ordinary differential geometry the Levi-Civita connection is uniquely determined when torsion is set to zero and metricity is satisfied, this is not so for the case of noncommutativity, as may be inferred from the discussion above.

We also consider interesting to remark at this point the relation between Lie algebras on fibers of Principal Fiber Bundles and covariant derivative diffeomorphisms on their base spaces and the corresponding Inner and Outer derivations in Noncommutative Geometry, as described above. This may also be seen when relating torsion in the noncommutative context to the perhaps more familiar definition in the classical differential geometry \cite{D-V3} by identifying it with 
the $\mathcal Z(\mathcal A)$-bilinear antisymmetric mapping  $T:Der(\mathcal A)\times Der(\mathcal A)\to Der(\mathcal A)$ as:
\be\label{tor4}
T(\hat X_{\mu},\hat X_{\nu})=\nabla_{X_{\mu}}(\hat X_{\nu})-\nabla_{X_{\nu}}(\hat X_{\mu})-[\hat X_{\mu},\hat X_{\nu}], \qquad \forall \hat X_{\mu},\hat X_{\nu} \in Der(\mathcal A),
\ee
and observing  that the metricity condition (\ref{met}) together with zero torsion are equivalent \cite{D-V} to a bilinear mapping  ${\frak g}$ of $Der(\mathcal A)\times Der(\mathcal A)$ into $\mathcal A$ such that:
\be\label{bilin}
 D_{ X_{\kappa}}{\frak g} (\hat X_{\mu}, \hat X_{\nu})=
  {\frak g}\big( \nabla_{ X_{\kappa}} (\hat X_{\mu}), \hat X_{\nu}\big)+ {\frak g}\big (\hat X_{\mu}, \nabla_{ X_{\kappa}}(\hat X_{\nu})\big).
\ee
 Therefore,  summing over the cyclic permutations of the later, while making use of (\ref{tor4}) and the symmetry 
of the metric components, one arrives at the following noncommutative expression for the connection symbols:
\begin{align}\label{L-C}
2{\frak g}(\nabla_{X_{\mu}}(\hat X_{\nu}),\hat X_{\kappa})=&\hat X_{\mu}({\frak g}(\hat X_{\nu},\hat X_{\kappa})) + \hat X_{\nu}({\frak g}(\hat X_{\mu},\hat X_{\kappa}))
-\hat X_{\kappa}({\frak g}(\hat X_{\mu},\hat X_{\nu}))+{\frak g} ([\hat X_{\mu}, \hat X_{\nu}], \hat X_{\kappa})\nonumber\\
&-{\frak g} ([\hat X_{\nu}, \hat X_{\kappa}], \hat X_{\mu})+{\frak g} ([\hat X_{\kappa}, \hat X_{\mu}], \hat X_{\nu}).
\end{align}

The above is formally  analogous to the one resulting in ordinary differential geometry based on a non-coordinate basis. Note however that the derivations here are Outer-derivations.

 Hence in order to proceed further in determining the connection symbols in these equations within the Levi-Civita context, we recall first that they are valued in the center of the algebra
$\mathcal Z(\mathcal A)$ and so are their derivations. Thus, making use of (\ref{subder4}) where $\mathcal{N}^{\sigma}_{\rho \mu}\in  {\mathcal{Z}}(\mathcal{A})$, so that
\begin{align}\label{derx0}
D_{X_{\rho}} X_{\mu} =&\mathcal{N}_{\rho\mu}^{\lambda}X_{\lambda},
\end{align}
 and by further substituting (\ref{derx0}) into
\begin{align}\label{bilin2}
D_{X_{\rho}}{\mathcal O_{1}}=\xi_{3}D_{X_{\rho}}X_0-\xi_{0}D_{X_{\rho}}X_3 ,
\end{align}
we obtain the relations
\begin{subequations} 
\begin{align}\label{4.23a}
\mathcal{N}_{\rho3}^{i}&=\xi_{3}\xi_{0}^{-1}\mathcal{N}_{\rho0}^{i},\;\; i=1,2,\\\label{4.23b}
\phi_{\rho}&= \mathcal{N}_{\rho0}^{0}-\xi_{0}\xi_{3}^{-1}\mathcal{N}_{\rho3}^{0}=
 \mathcal{N}_{\rho3}^{3}-\xi_{3}\xi_{0}^{-1}\mathcal{N}_{\rho0}^{3},
\end{align}
\end{subequations}
from where
\be\label{bilin4}
D_{X_{\rho}}{\mathcal O_{1}}= \phi_{\rho}(\mathcal Z(\mathcal A)){\mathcal O_{1}}. 
 \ee
 Next, in analogy with (\ref{bilin2}) we have
\begin{equation}\label{4.24}
D_{X_{\rho}}{\mathcal O_{2}}=D_{X_{\rho}}\big((X_1)^2 +(X_2)^2\big)=\sum_{k=1}^{2} \big (D_{X_{\rho}} X_k\big) X_k+\sum_{k=1}^{2} X_{k} \big (D_{X_{\rho}} X_k\big), 
\end{equation}
that yields the relations
\be\label{4.25}
 \mathcal{N}^{1}_{\rho1}=\mathcal{N}^{2}_{\rho2}, \;\; \mathcal{N}^{2}_{\rho1}=-\mathcal{N}^{1}_{\rho2},\;\;\mathcal{N}^{0}_{\rho1}=\mathcal{N}^{3}_{\rho1}
 =\mathcal{N}^{0}_{\rho2}=\mathcal{N}^{3}_{\rho2}=0,
\ee
 so (\ref{derx0}) for $\mu=1,2$ becomes
\begin{align}\label{der8}
D_{X_{\rho}} X_{1} =\mathcal{N}_{\rho1}^{1}X_1+\mathcal{N}_{\rho1}^{2}X_2, \qquad\qquad 
D_{X_{\rho}} X_{2} =-\mathcal{N}_{\rho1}^{2}X_1+\mathcal{N}_{\rho1}^{1}X_2,
\end{align}
from where 
\begin{equation}\label{der9}
D_{X_{\rho}}{\mathcal O_{2}}=2\mathcal{N}_{\rho1}^{1}{\mathcal O_{2}}.
\end{equation}
Moreover, from the Leibnitz rule and, by acting with the derivations on the algebra commutator ({\it i.e.} on  $[X_{\mu},X_{\nu}]=\mathcal{C}_{\mu\nu}^{\lambda}X_{\lambda}$) we have that 
\begin{align}\label{leibnitzrule1}
[X_{\mu},  D_{X_{\rho}} X_{\nu}]+[D_{X_{\rho}}X_{\mu},X_{\nu}]=\mathcal{C}_{\mu\nu}^{\lambda}
D_{X_{\rho}}X_{\lambda},
\end{align}
which when making use again of (\ref{derx0}) leads to the conditions
\be\label{leib5}
C^{\tau}_{\mu\sigma}\mathcal{N}^{\sigma}_{\rho\nu}+
C^{\tau}_{\sigma\nu}\mathcal{N}^{\sigma}_{\rho\mu}-C^{\sigma}_{\mu\nu}\mathcal{N}^{\tau}_{\rho\sigma}=0.
\ee
 These, in turn, imply that $\mathcal{N}_{\rho0}^{3}=-\xi_{0}\xi_{3}^{-1}\mathcal{N}_{\rho0}^{0}$ and $\mathcal{N}_{\rho3}^{3}=-\xi_{0}\xi_{3}^{-1}\mathcal{N}_{\rho3}^{0}$. Note however that when substituting these relations into
the second equality in (\ref{4.23b}) results in its first equality. Thus the independent relations from those two sets are 
\be\label{leib3a}
\mathcal{N}^{3}_{\rho\tau}=-\xi_{0}\xi_{3}^{-1}
\mathcal{N}^{0}_{\rho\tau},\;\; \tau=0,3.
\ee
Therefore for $\mu=0,3$ the equation (\ref{derx0}) takes the form
\begin{align}\nonumber
D_{X_{\rho}}X_{0}&=\xi_{3}^{-1}\mathcal{O}_{1}\mathcal{N}_{\rho0}^{0}+
\mathcal{N}_{\rho0}^{1}X_{1}+\mathcal{N}_{\rho0}^{2}X_{2},\\\label{derx0-1} D_{X_{\rho}}X_{3}&=\xi_{3}^{-1}\mathcal{O}_{1}\mathcal{N}_{\rho3}^{0}+
\xi_{3}\xi_{0}^{-1}(\mathcal{N}_{\rho0}^{1}X_{1}+\mathcal{N}_{\rho0}^{2}X_{2})
\end{align}

Finally, it would be apparently reasonable that additional relations would result from a double commutator polynomial of the form
$$[{X}_{\mu},[{X}_{\nu},D_{{X}_{\rho}}{X}_{\lambda}]]+[{X}_{\nu},[D_{{X}_{\rho}}{X}_{\lambda},{X}_{\mu}]]+[D_{{X}_{\rho}}{X}_{\lambda},[{X}_{\mu},{X}_{\nu}]].$$
Nonetheless  making use once more of (\ref{derx0}) and noting that
\begin{align}\label{1}
[{X}_{\mu},[{X}_{\nu},D_{{X}_{\rho}}{X}_{\lambda}]]=[{X}_{\mu},[{X}_{\nu},\mathcal{N}_{\rho\lambda}^{\sigma}{X}_{\sigma}]]=\mathcal{N}_{\rho\lambda}^{\sigma}[{X}_{\mu},[{X}_{\nu},{X}_{\sigma}]]
\end{align} 
and observing also that the commutator in the second equality has to satisfy the Jacobi identity, we have that
\begin{align*}
\mathcal{N}_{\rho\lambda}^{\sigma}[{X}_{\mu},[X_{\nu},{X}_{\sigma}]]&=-\mathcal{N}_{\rho\lambda}^{\sigma}([{X}_{\nu},[{X}_{\sigma},{X}_{\mu}]]+[{X}_{\sigma},[{X}_{\mu},{X}_{\nu}]])\\
&=-[{X}_{\nu},[\mathcal{N}_{\rho\lambda}^{\sigma}{X}_{\sigma},{X}_{\mu}]]-[\mathcal{N}_{\rho\lambda}^{\sigma}{X}_{\sigma},[{X}_{\mu},{X}_{\nu}]]\\
&=-[{X}_{\nu},[D_{\hat{X}_{\rho}}{X}_{\lambda},{X}_{\mu}]]-[D_{{X}_{\rho}}{X}_{\lambda},[{X}_{\mu},{X}_{\nu}]].
\end{align*}   
Hence
\be\label{der10}
[{X}_{\mu},[{X}_{\nu},D_{{X}_{\rho}}{X}_{\lambda}]]+[{X}_{\nu},[D_{{X}_{\rho}}{X}_{\lambda},{X}_{\mu}]]+[D_{{X}_{\rho}}{X}_{\lambda},[{X}_{\mu},{X}_{\nu}]]=0,
\ee
and this in turn implies the condition 
\be\label{relfin}
\mathcal{N}^{\sigma}_{\rho\lambda}\Big ( C^{\kappa}_{\nu\sigma}  C^{\tau}_{\mu\kappa}+  C^{\kappa}_{\sigma\mu} C^{\tau}_{\nu\kappa} + C^{\kappa}_{\mu\nu}  C^{\tau}_{\sigma\kappa}\Big)=0,
\ee
which as an identity leads to no additional relations between the components of the $\mathcal{N}^{\lambda}_{\rho\mu}$.

We still have one more condition that comes from the explicit application of (\ref{z-mod1}) over the algebra elements $X_{\lambda}$ by means of (\ref{derx0}) which results in
\begin{equation}\label{constriction-1}
D_{X_{\mu}}(\mathcal{N}_{\nu\lambda}^{\tau})-D_{X_{\nu}}(\mathcal{N}_{\mu\lambda}^{\tau})+\mathcal{N}_{\nu\lambda}^{\sigma}
\mathcal{N}_{\mu\sigma}^{\tau}-\mathcal{N}_{\mu\lambda}^{\sigma}
\mathcal{N}_{\nu\sigma}^{\tau}=B_{[\mu,\nu]}^{\sigma}\mathcal{N}_{\sigma\lambda}^{\tau}.
\end{equation}
It is fairly straightforward to obtain the conditions implied by this equation by making use of (\ref{4.25}), (\ref{leib3a}) and (\ref{B10b}), we thus get
\begin{subequations}
\begin{align}\label{const6a}
&D_{X_{\mu}}\big(\mathcal {N}_{\nu 0}^{0}\big)-D_{X_{\nu}}\big(\mathcal {N}_{\mu 0}^{0})=B^{\sigma}_{[\mu,\nu]}\mathcal {N}_{\sigma 0}^{0},\\\label{const6b}
&D_{X_{\mu}}(\mathcal{N}_{\nu0}^{1})-D_{X_{\nu}}(\mathcal{N}_{\mu0}^{1})+\mathcal{N}_{\nu0}^{1}\mathcal{N}_{\mu1}^{1}+\mathcal{N}_{\nu0}^{2}\mathcal{N}_{\mu2}^{1}-\mathcal{N}_{\mu0}^{1}\mathcal{N}_{\nu1}^{1}
-\mathcal{N}_{\mu0}^{2}\mathcal{N}_{\nu2}^{1}=B_{[\mu,\nu]}^{\sigma}\mathcal{N}_{\sigma0}^{1},\\\label{const6c}
&D_{X_{\mu}}(\mathcal{N}_{\nu1}^{1})-D_{X_{\nu}}(\mathcal{N}_{\mu1}^{1})=B_{[\mu,\nu]}^{\sigma}\mathcal{N}_{\sigma1}^{1},\\\label{const6d}
&D_{X_{\mu}}(\mathcal{N}_{\nu2}^{1})-D_{X_{\nu}}(\mathcal{N}_{\mu2}^{1})=B_{[\mu,\nu]}^{\sigma}\mathcal{N}_{\sigma2}^{1},\\\label{const6e}
&D_{X_{\mu}}(\mathcal{N}_{\nu0}^{2})-D_{X_{\nu}}(\mathcal{N}_{\mu0}^{2})+\mathcal{N}_{\nu0}^{1}\mathcal{N}_{\mu1}^{2}+\mathcal{N}_{\nu0}^{2}\mathcal{N}_{\mu2}^{2}-\mathcal{N}_{\mu0}^{1}\mathcal{N}_{\nu1}^{2}
-\mathcal{N}_{\mu0}^{2}\mathcal{N}_{\nu2}^{2}=B_{[\mu,\nu]}^{\sigma}\mathcal{N}_{\sigma0}^{2}.
\end{align}
\end{subequations}
We can now apply the above results and expressions to the metricity condition in (\ref{met}) to get explicit expressions for the action of the derivations on the metric components previously displayed in (\ref{qmet}). Thus recalling 
(\ref{linder1}) we have
 \begin{align}\label{met5}
 0=\nabla_{ X_{\rho}} \big( {\frak g}_{\mu\nu} \; \omega^{\mu}\otimes\omega^{\nu}\big)&=(D_{X_{\rho}}{\frak g}_{\mu\nu})\omega^\mu \otimes \omega^{\nu}+{\frak g}_{\mu\nu}\nabla_{ X_{\rho}}\omega^\mu \otimes \omega^{\nu}+ {\frak g}_{\mu\nu}\omega^\mu \otimes \nabla_{ X_{\rho}}\omega^{\nu},
\end{align}
and further making use of (\ref{cons3d}) we get:
\be\label{met6}
D_{X_{\rho}}{\frak g}_{\mu\nu}- 
{\frak g}_{\sigma\nu}\Gamma^{\sigma}_{\rho\mu}-{\frak g}_{\mu\sigma}\Gamma^{\sigma}_{\rho\nu}=0.
\ee
From here we can get an explicit expression for the connection symbols by first taking cyclic permutations of the free three lower indices on the terms in (\ref{met6}), and making use of (\ref{tor3}), yields
\be\label{met7}
2{\frak g}_{\sigma\nu}\Gamma^{\sigma}_{\rho\mu}= D_{X_{\rho}}{\frak g}_{\mu\nu}+D_{X_{\mu}}{\frak g}_{\rho\nu}-D_{X_{\nu}}{\frak g}_{\rho\mu}-{\frak g}_{\sigma\nu}B_{[\mu\rho]}^{\sigma}+
{\frak g}_{\sigma\mu}B_{[\nu\rho]}^{\sigma}+{\frak g}_{\rho\sigma}B_{[\nu\mu]}^{\sigma}.
\ee
In addition to the fact that this result makes more evident our previous remarks on the equivalence between the formalism leading to (\ref{L-C}) and the above approach, we could arrive at a more manageable approach for considering explicit scenarios for 
calculating the 64 $\Gamma^{\mu}_{\rho\kappa}$'s; either by making use of the left or right inverse metric matrix discussed in Section (3.3) or by an alternate approach that would make use of the previously derived expressions of the metric components as shown in (\ref{qmet}), together with the equations for the derivations in (\ref{der8}), (\ref{derx0-1}), followed by a match of the resulting monomials on each side of (\ref{met6}).

Furthermore, note that in order to preserve the metricity condition we need the derivations of the metric components explicitly. We can obtain them by using the derivations (\ref{der8}) and (\ref{derx0-1}), followed by a match of the resulting monomials on each side of (\ref{met6}) (see the Appendix B). This results in the following relations that the connection symbols must satisfy:  
\begin{subequations}
\begin{align}\label{gammas1-1}
&\Gamma_{\rho1}^{0}=\Gamma_{\rho2}^{0}=0,\\\label{gammas2-1}
&\Gamma_{\rho1}^{1}=\Gamma_{\rho2}^{2}=\mathcal{N}_{\rho1}^{1},\\\label{gammas3-1}
&\Gamma_{\rho3}^{1}=\xi_{3}\xi_{0}^{-1}\Gamma_{\rho0}^{1}=\xi_{3}\xi_{0}^{-1}\mathcal{N}_{\rho0}^{1},\\\label{gammas4-1}
&\Gamma_{\rho1}^{2}=-\Gamma_{\rho2}^{1}=\mathcal{N}_{\rho1}^{2},\\\label{gammas5-1}
&\Gamma_{\rho3}^{2}=\xi_{3}\xi_{0}^{-1}\Gamma_{\rho0}^{2}=\xi_{3}\xi_{0}^{-1}\mathcal{N}_{\rho0}^{2},\\\label{gammas6-1}
&\Gamma_{\rho0}^{3}=-\xi_{0}\xi_{3}^{-1}\Gamma_{\rho0}^{0}=-\xi_{0}\xi_{3}^{-1}l_{2}\mathcal{O}_{1}\mathcal{N}_{\rho0}^{0},\\\label{gammas7-1}
&\Gamma_{\rho1}^{3}=\Gamma_{\rho2}^{3}=0,\\\label{gammas8-1}
&\Gamma_{\rho3}^{3}=-\xi_{0}\xi_{3}^{-1}\Gamma_{\rho3}^{0}=-\xi_{0}\xi_{3}^{-1}l_{3}\mathcal{O}_{1}\mathcal{N}_{\rho0}^{0},
\end{align}
\end{subequations}
with the additional constraint on the so far undetermined coefficients of the metric components:
\begin{subequations}
\begin{align}\label{constriction1}
&a_{1}=\xi_{0}\xi_{3}^{-1}b_{1}=\xi_{0}^{2}\xi_{3}^{-2}c_{1}.
\end{align}
\end{subequations}

On the other hand, by making use of the torsionless condition (\ref{tor3}) and the equations (\ref{gammas1-1}-\ref{gammas8-1}) derived in the Appendix B, we obtain that the $B's$ are given by:
\begin{subequations}
\begin{align}\label{B1}
B_{[\mu,\rho]}^{1}&=\mathcal{N}_{\mu\rho}^{1}-\mathcal{N}_{\rho \mu}^{1},\qquad \mu\neq \rho,\qquad \mu, \rho=0,1,2,3\\\label{B2}
B_{[\mu,\rho]}^{2}&=\mathcal{N}_{\mu\rho}^{2}-\mathcal{N}_{\rho \mu}^{2},\\\label{B3}
B_{[1,2]}^{3}&=B_{[1,2]}^{0}=0,\\\label{B4}
B_{[0,\rho]}^{3}&=\xi_{0}\xi_{3}^{-1} B_{[\rho,0]}^{0}=\xi_{0}\xi_{3}^{-1}l_{2}\mathcal{O}_{1}
\mathcal{N}_{\rho0}^{0}-\xi_{0}\xi_{3}^{-1}l_{3}\mathcal{O}_{1}
\mathcal{N}_{00}^{0}\delta_{\rho}^{3},\\\label{B5}
B_{[1,\rho]}^{3}&=-\xi_{0}\xi_{3}^{-1}B_{[1,\rho]}^{0}=-\xi_{0}\xi_{3}^{-1}(l_{2}\mathcal{O}_{1}\delta_{\rho}^{0}\mathcal{N}_{10}^{0}+
l_{3}\mathcal{O}_{1}\delta_{\rho}^{3}\mathcal{N}_{10}^{0}),\\\label{B6}
B_{[2,\rho]}^{3}&=-\xi_{0}\xi_{3}^{-1}B_{[2,\rho]}^{0}=-\xi_{0}\xi_{3}^{-1}(l_{2}\mathcal{O}_{1}\delta_{\rho}^{0}\mathcal{N}_{20}^{0}+
l_{3}\mathcal{O}_{1}\delta_{\rho}^{3}\mathcal{N}_{20}^{0}),\\\label{B7}
B_{[3,\rho]}^{3}&=\xi_{0}\xi_{3}^{-1} B_{[\rho,3]}^{0}=\xi_{0}\xi_{3}^{-1}l_{3}\mathcal{O}_{1}
\mathcal{N}_{\rho0}^{0}-\xi_{0}\xi_{3}^{-1}l_{2}\mathcal{O}_{1}
\mathcal{N}_{30}^{0}\delta_{\rho}^{3}.
\end{align}
\end{subequations}
From the above it readily follows that
\begin{align}
\begin{split}
&(\ref{B4})~\Longrightarrow~\mathcal{N}_{00}^{0}=0,\\
&(\ref{B7})~\Longrightarrow~(l_{3}-l_{2})\mathcal{N}_{30}^{0}=0,\\
&(\ref{B10a})~\Longrightarrow~\mathcal{N}_{03}^{0}=0,\\
&(\ref{B10b})~\Longrightarrow~\phi_{\rho}=l_{1}\mathcal{N}_{\rho0}^{0}.
\end{split}
\end{align}
Moreover, since the $\Gamma$'s are neither symmetric nor antisymmetric in the lower indices we make use of (\ref{tor3}) and (\ref{B1}-\ref{B6}) to get for components with indexes exchanged: 
\begin{subequations}
\begin{align}\label{antigammas1}
&\Gamma_{\mu\rho}^{1}=\mathcal{N}_{\mu\rho}^{1},\\\label{antigammas2}
&\Gamma_{\mu\rho}^{2}=\mathcal{N}_{\mu\rho}^{2},\\\label{antigammas3}
&\Gamma_{0\rho}^{3}=\Gamma_{0\rho}^{0}=0,\\\label{antigammas31}
&\Gamma_{1\rho}^{3}=-\Gamma_{1\rho}^{0}=-\xi_{0}\xi_{3}^{-1}(l_{2}
\mathcal{O}_{1}\mathcal{N}_{10}^{0}\delta_{\rho}^{0}+l_{3}
\mathcal{O}_{1}\mathcal{N}_{10}^{0}\delta_{\rho}^{3}),\\\label{antigammas4}
&\Gamma_{2\rho}^{3}=-\Gamma_{2\rho}^{0}=-\xi_{0}\xi_{3}^{-1}(l_{2}
\mathcal{O}_{1}\mathcal{N}_{20}^{0}\delta_{\rho}^{0}+l_{3}
\mathcal{O}_{1}\mathcal{N}_{20}^{0}\delta_{\rho}^{3}),\\\label{antigammas5}
&\Gamma_{3\rho}^{3}=-\xi_{0}\xi_{3}^{-1}\Gamma_{3\rho}^{0}=-\xi_{0}\xi_{3}^{-1}l_{2}
\mathcal{O}_{1}\mathcal{N}_{30}^{0}\delta_{\rho}^{3};
\end{align}
\end{subequations}
these relations, together with
\begin{align}
\begin{split}
&(4.20a)~\Longrightarrow~\mathcal{N}_{\rho3}^{i}=\xi_{3}\xi_{0}^{-1}\mathcal{N}_{\rho 0}^{i},\;\; i=1,2,\\
&(4.20b)~\Longrightarrow~\phi_{\rho}= \mathcal{N}_{\rho0}^{0}-\xi_{0}\xi_{3}^{-1}\mathcal{N}_{\rho3}^{0}=\mathcal{N}_{\rho3}^{3}-\xi_{3}\xi_{0}^{-1}\mathcal{N}_{\rho0}^{3},\\
&(4.25)~\Longrightarrow~\mathcal{N}^{1}_{\rho1}=\mathcal{N}^{2}_{\rho2}, \;\; \mathcal{N}^{2}_{\rho1}=-\mathcal{N}^{1}_{\rho2},\;\;\mathcal{N}^{0}_{\rho1}=\mathcal{N}^{3}_{\rho1}
 =\mathcal{N}^{0}_{\rho2}=\mathcal{N}^{3}_{\rho 2}=0,\\
&(\ref{leib3a})\mathcal{N}^{3}_{\rho\tau}=-\xi_{0}\xi_{3}^{-1}\mathcal{N}^{0}_{\rho\tau},\;\; \tau=0,3,
\end{split}
\end{align}
and (see Appendix B)
\begin{subequations}
\begin{align}\label{els}
l&=\gamma_{3}+2\xi_{0}\xi_{3}^{-1}\gamma_{4},\\\label{13}
l_{1}&=2a_{2}(2a_{2}-\xi_{0}\xi_{3}^{-1}l)^{-1},\\\label{14}
l_{2}(\mathcal{O}_{1})&=\frac{1}{2}(\xi_{3}a_{0}-b_{0}\xi_{0})^{-1}[(\gamma_{1}+2\gamma_{4}\xi_{3}^{-1}\mathcal{O}_{1})l_{1}+a_{1}l(\xi_{0}\xi_{3}^{-1}l-2a_{2})^{-1}],\\\label{15}
l_{3}(\mathcal{O}_{1})&=\frac{1}{2}(\xi_{3}b_{0}-c_{0}\xi_{0})^{-1}[(\kappa_{1}+2\kappa_{4}\xi_{3}^{-1}\mathcal{O}_{1})l_{1}+\xi_{0}^{-2}\xi_{3}^{2}a_{1}l(\xi_{0}\xi_{3}^{-1}l-2a_{2})^{-1}],
\end{align}
\end{subequations}
which were previously derived along this section and summarized here, constitute the basic material required for our next discussion.

\subsection{The Pseudo-Riemannian Curvature of $\nabla_{X}$}
 	
Recall now that in a general basis the Pseudo-Riemann curvature of $\nabla$ is the bilinear antisymmetric mapping   
\be\label{curv11}
(\hat X_\mu ,\hat X_\nu )\mapsto R(\hat X_\mu , \hat X_\nu)= (\nabla_{X_\mu}\circ \nabla_{ X_\nu} - \nabla_{ X_\nu}\circ \nabla_{ X_\mu}-\nabla_{[ D_{X_\mu}, D_{ X_\nu}] }),
\ee
where $\nabla_{X_{\rho}}$ is as defined in (\ref{linder2}).
  	
The Pseudo-Riemann Curvature tensor is the given by
\be\label{riem1}
R(\hat X_\mu ,\hat X_\nu, \omega^\sigma, \hat X_\rho ):= i_{\hat X_\rho }\triangleright\big(R(\hat X_\mu ,\hat X_\nu)\triangleright\omega^\sigma \big) .
\ee
When substituting the covariant derivation (\ref{cons3d}) and (\ref{tor3}) in the above expression it immediately follows that the quantum Non-Commutative Pseudo-Riemannian tensor components are given by
\be\label{riem2}
R_{\rho\nu\mu}^{\;\sigma}:=- D_{ X_\mu}(\Gamma_{\nu\rho}^{\sigma}) + D_{ X_\nu}(\Gamma_{\mu\rho}^{\sigma}) - \Gamma_{\nu\rho}^{\lambda}\Gamma_{\mu\lambda}^{\sigma}+\Gamma_{\mu\rho}^{\lambda}\Gamma_{\nu\lambda}^{\sigma}+\Gamma_{\lambda\rho}^{\sigma}(\Gamma_{\mu\nu}^{\lambda}-\Gamma_{\nu\mu}^{\lambda}),
\ee
where the curvature is clearly antisymmetric in the $\mu,\nu$ indices.\\
Now in order to get explicit relations for the terms on the right hand of the above equation, recall first that the exterior derivations are Lie derivations defined by (\ref{lie20}) and making use of (\ref{bilin4}) and (\ref{der9})
we have that
\begin{align}\label{curv1}
D_{ X_\mu}(\Gamma_{\nu\rho}^{\sigma})=& i_{\hat X_{\mu}}\circ d\;\Gamma_{\nu\rho}^{\sigma}(\mathcal O_1 ,\mathcal O_2 ) =  i_{\hat X_{\mu}}\circ \Big(\frac{\partial \Gamma_{\nu\rho}^{\sigma}}{\partial \mathcal O_1} d{\mathcal O_1} +\frac{\partial \Gamma_{\nu\rho}^{\sigma}}{\partial \mathcal O_2}d{\mathcal O_2} \Big)\nonumber\\
=&\frac{\partial \Gamma_{\nu\rho}^{\sigma}}{\partial \mathcal O_1} D_{X_{\mu}}{\mathcal O_1} +\frac{\partial \Gamma_{\nu\rho}^{\sigma}}{\partial \mathcal O_2}D_{X_{\mu}}{\mathcal O_2} = \frac{\partial \Gamma_{\nu\rho}^{\sigma}}{\partial \mathcal O_1} \phi_{\mu}(\mathcal Z){\mathcal O_1}+ 2{\frac{\partial \Gamma_{\nu\rho}^{\sigma}}{\partial \mathcal O_2} \mathcal N^{1}_{\mu 1}}{\mathcal O_2}.
\end{align}
Consequently (\ref{riem2}) reads
\be\label{curv2}
R^{\sigma}_{\rho\nu\mu}= \frac{\partial \Gamma_{\mu\rho}^{\sigma}}{\partial \mathcal O_1} \phi_{\nu}(\mathcal Z){\mathcal O_1}+ 2{\frac{\partial \Gamma_{\mu\rho}^{\sigma}}{\partial \mathcal O_2} \mathcal N^{1}_{\nu 1}}{\mathcal O_2}-\frac{\partial \Gamma_{\nu\rho}^{\sigma}}{\partial \mathcal O_1} \phi_{\mu}(\mathcal Z){\mathcal O_1}- 2{\frac{\partial \Gamma_{\nu\rho}^{\sigma}}{\partial \mathcal O_2} \mathcal N^{1}_{\mu 1}}{\mathcal O_2}  - \Gamma_{\nu\rho}^{\lambda}\Gamma_{\mu\lambda}^{\sigma}+\Gamma_{\mu\rho}^{\lambda}\Gamma_{\nu\lambda}^{\sigma}+\Gamma_{\lambda\rho}^{\sigma}B_{[\mu\nu]}^{\lambda},
\ee
and, making use of (\ref{tor3}), the quantum Ricci tensor is given by
\be\label{riem3}
R_{\rho\mu}=\frac{\partial \Gamma_{\mu\rho}^{\sigma}}{\partial \mathcal O_1} \phi_{\sigma}(\mathcal Z){\mathcal O_1}+ 2{\frac{\partial \Gamma_{\mu\rho}^{\sigma}}{\partial \mathcal O_2} \mathcal N^{1}_{\sigma 1}}{\mathcal O_2}-\frac{\partial \Gamma_{\sigma\rho}^{\sigma}}{\partial \mathcal O_1} \phi_{\mu}(\mathcal Z){\mathcal O_1}- 2{\frac{\partial \Gamma_{\sigma\rho}^{\sigma}}{\partial \mathcal O_2} \mathcal N^{1}_{\mu 1}}{\mathcal O_2} +\Gamma_{\mu\rho}^{\lambda}\Gamma_{\sigma\lambda}^{\sigma}-\Gamma_{\lambda\rho}^{\sigma}\Gamma_{\sigma\mu}^{\lambda}.\\
\ee
To obtain a scalar from the above, would be apparently simple by tracing with the inverses of the metric components derived in the previous section. Recall however that even-though the inverses commute with (\ref{riem3}), due to the fact that all its Lie-algebraic entries are in the center of the algebra, the tensor $R_{\rho\mu}$ is not symmetric since, as we have shown previously, the connection symbols are neither symmetric nor antisymmetric. Thus acting with the inverses from the right or from the left, as shown in (\ref{inv6} and \ref{inv7}), leads to different traces. This is of course a reflection of the intrinsic noncommutativity of the quantum geometry of the problem and, although it is not our goal here to delve into the detailed quantization of the above field equations, we consider interesting to comment next on some possible implications of our results as well as possible further developments based on our results.

\section {Quantization}
Assume now that the covariant connection in (\ref{covder1}) is compatible with a Hermitian-Hilbert structure. We can then consider the quantization of the quantum curvature (\ref{riem3}) or the quantization of the Poincar\'e Lie algebra space itself. For the first case we shall describe the essentials of Radial Quantization (see {\it e.g.} \cite{papa} ), based on the Hilbert space and quantum mechanical evolution. 

\subsection{Radial quantization}

According to the algebraic structure given in (\ref{lie}), the center of the algebra $ {\mathcal Z(\mathcal A)}$ can be written in terms of the operators  ${\mathcal O_1}, {\mathcal O_2}$. Moreover, since ${\mathcal O_2}$ is a positive definite operator we can consider that it can be chosen to play the role of the time operator in our geometry. Namely, we will use a foliation of $S^1$ spheres of different radii. Usually this procedure is 
called the ``radial quantization". In this context, we assume that the center of the  circle is located at $X_{1}=0, X_{2} =0$, since according to our algebra these coordinates commute. However, in the true four dimensional space generated by the algebra, this center is not well defined as the $X_{1}$ and $X_{2}$ coordinates fluctuate with respect to the $X_{0}$ and $X_{3}$ coordinates.

 We note though, that we could equally well quantize our geometry with respect to any other point and that this should give the same physics. 
Now, the generator that moves us from one circle to the other in the radial quantization formalism is the dilatation generator ${\mathcal D}$, that in our case will be the conjugate variable  associated to      
${\mathcal O_2}$. It will thus play the role of the Hamiltonian which, in the context of Noncommutative Gravity, defines the time-lapse function $N$. If we further recall that our generalization of vector fields corresponds to the derivations $D_{X_\mu}$ and, taking into account that in classical mechanics $H=-p_{t}$, then in our case the corresponding Hamiltonian operator will be
\begin{equation}\label{ham1}
\hat H=iD_{\mathcal O_2}.
\end{equation}
In addition, the states living on the circles will be classified according to their scaling dimension 
\begin{equation}\label{ham2}
\hat H |\Delta \rangle =i\Delta |\Delta\rangle,
\end{equation}
and since the only generator that commutes with $\hat H$ is the momentum conjugate to ${\mathcal O_1}$, which in our case will be $\hat M_{\mathcal O_1}=-iD_{\mathcal O_1}$, we get that our states will be classified by
\begin{equation}\label{ham3}
\hat M_{\mathcal O_1} |\Delta, m\rangle = m_{\mathcal O_1} |\Delta, m\rangle.
\end{equation}
Notice that the eigenvalues of $\hat H$ and $\hat M_{\mathcal O_1}$ are continuous since we don't have a quantization condition.
Furthermore, from (\ref{riem3}) our quantum Ricci tensor is defined only in terms of ${\mathcal O_1}, {\mathcal O_2}$,  then acting on the basis defined by (\ref{ham3}) we will get a derivative action. Under this action, the
resulting classical space will have the topology of a cylinder with the universe expanding in the radial time direction.\\

\begin{figure}
 \centering
 \includegraphics[width=16cm, height=10cm]{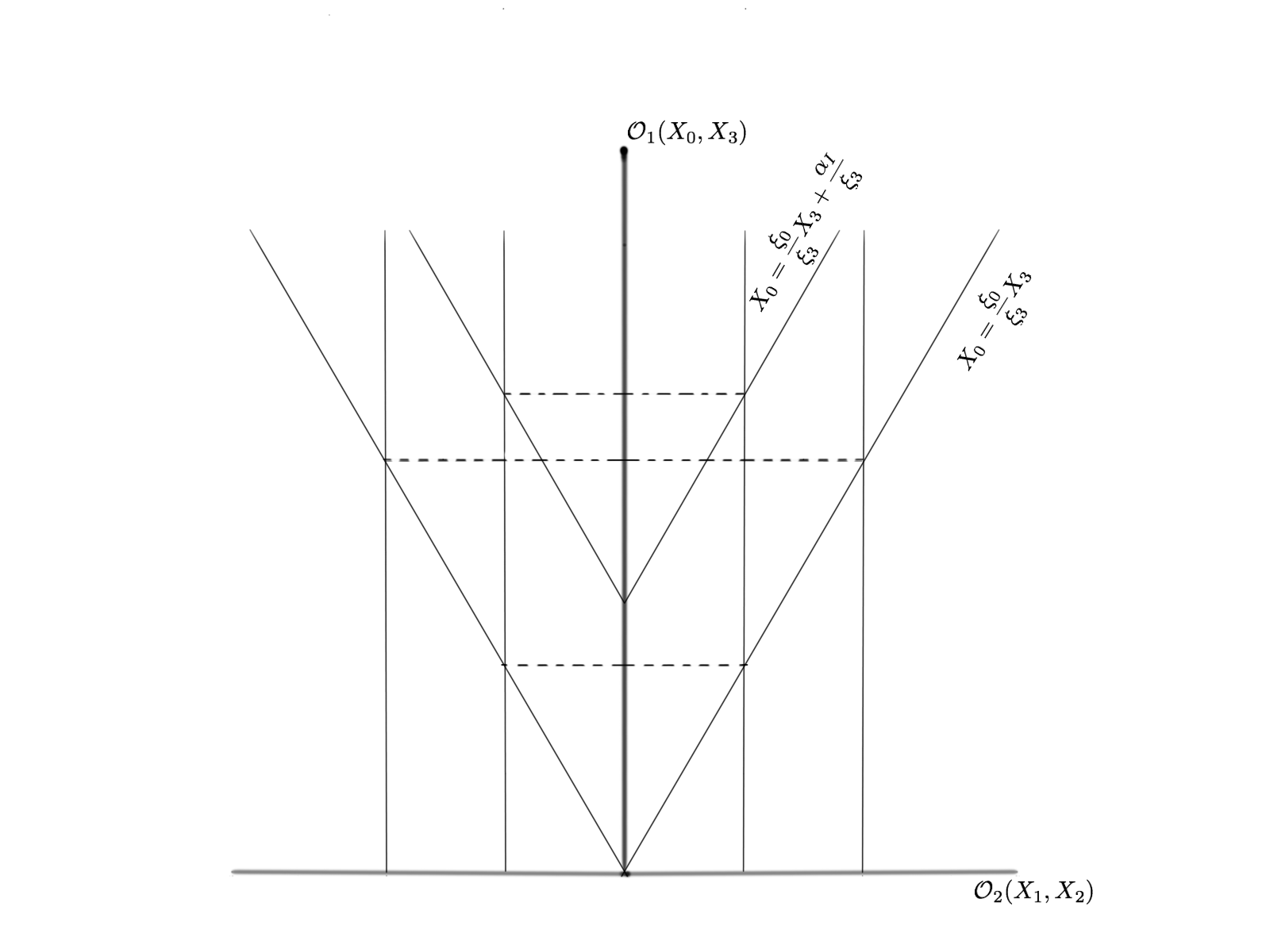}
 \caption{A representation of the upper-half quantum space}
\end{figure}

\subsection{Quantization of the Geometry}

 From the algebraic point of view, given the algebra (\ref{lie}) we can, together with the ${\mathcal O_1}, {\mathcal O_2}$ operators, introduce the operators  $X_3$ and $X_0$ in order to conform a complete set of commuting operators and eigenkets $| o_{1}, o_{2}, j, k >$, where the entries in the bracket are the eigenvalues of the corresponding quantum operators, on which the Hilbert space will be based.  
Consider now the so-called ladder operators $X_\pm=X_1 \pm i X_2$ which satisfy the algebra
\begin{eqnarray}
[X_+,X_-]=0,\qquad [X_3, X_\pm]=\pm 2\xi_3 X_\pm, \qquad  [X_0, X_\pm]=\pm 2\xi_0 X_\pm,\end{eqnarray} acting as raising and lowering operators for the eigenvalues $j, k$, of $X_3$ and $X_0$  respectively. Thus
\begin{align}
X_3(X_\pm)| o_{1}, o_{2}, j, k> &=(j\pm2\xi_3)(X_\pm)| o_{1}, o_{2}, j, k>,\nonumber\\ 
X_0(X_\pm)| o_{1}, o_{2}, j, k> &=(k\pm2\xi_0)(X_\pm)| o_{1}, o_{2}, j, k>, 
\end{align}
so the kets $(X_\pm)| o_{1}, o_{2}, j, k>$ are simultaneous eigenvectors of $X_3$ and $X_0$, with eigenvalues $(j\pm2\xi_3)$ and $(k\pm2\xi_0)$.
Note, however, that since $X_+X_-=X_-X_+={\mathcal O_2}$ the only condition that we get for the eigenvalues of  ${\mathcal O_2}$ is $o_2\geq 0$. It is interesting to note that the operators $X_\pm$ induce a lattice structure on the geometry, with a of length $2\xi_3$ in the $X_3$ direction and  $2\xi_0$ in the $X_0$ direction. This structure \cite{asht} is similar to the one introduced in Polymer Quantum Mechanics. However, in this case the lattice is induced by the noncommutativity, whereas in the case of Polymer Quantum Mechanics it appears as discrete holonomies of fiber bundles in an effective theory evolving from what it will be an ontologically more fundamental Quantum Gravity.

Moreover, consider now the diagram  in Figure 1 above, where we have tried to represent schematically the 4-dimensional space generated by the Lie-algebra considered here. 
 The height of the ordinates is determined by the value of $\mathcal O_1$ in the $(X_{0},X_{3})$ hyperplane,  while the width of the abscissas represent the diameter of the cylinders determined by the value of $\mathcal O_2$ in the $(X_{1},X_{2})$ hyperplane. Noting also that $\mathcal O_{1}=0 \Rightarrow X_{0}=\frac{\xi_0}{\xi_3}X_{3}$ and $\mathcal O_{1}>0 \Rightarrow X_{0}=\frac{\xi_0}{\xi_3}X_{3}+ \frac{\alpha_{I}}{\xi_3}$, where $\alpha_{I} $ has dimensions of $\lambda_{P}^{2}$ (the square of the Planck length) and, although at this point it would be possible to take it as an uncountable set we could, by a Bayesian reasoning consistent with the arguments in the Introduction, take it so that the magnitude of the quotient $\frac{\alpha_{I}}{\xi_3}$ be also of the order of a Planck length. Hence, since the elements of these two hyperplanes {\bf do not commute}, the intersections of those equations of constant slope that determine the different ``points" of crossing  of the $(X_{1},X_{2})$ and $(X_{0}, X_{3})$ hyperplanes, are actually fuzzy points over discrete intervals.
Note in particular, that this applies to the origin, where $\mathcal O_{1}=0,\;\; \mathcal O_{2}=0 $, and in this sense 
noncommutativity solves the singularity problem of space-time and induces the possibility of discreteness of the eigenvalues of (\ref{ham2}-\ref{ham3}).

\section{The Twisted Metric Deformation Revisited}

The main objective of this section is to relate our results of Secs. 3 and 4 to the twisted deformation formalism discussed in Sec.2, in order to obtain the Weyl symbols for the quantum metric to further investigate the formula for the curvature in (\ref {riem2}). For this purpose we recall that the Weyl symbol of a product of operators $(\hat{f}\hat{g})_{W}$ is the twisted product of their associated Weyl symbols, ({\it c.f.}(\ref{twistzeta2}))
\begin{equation}\label{5.1}
\frak f_{W}\star\frak g_{W}=m\circ\Big(\exp{(i\xi^{\lambda}\partial_{\lambda}\wedge \overline{M}_{\alpha\beta})}\frak f_{W}\otimes\frak g_{W}\Big), \qquad \overline{M}_{\alpha\beta}:=\alpha\partial_\beta - x_{\beta}\partial_\alpha.
\end{equation}

Thus, to derive the Weyl symbols of the quantum metrics ({\it cf.} equations (\ref{Coeff1a}-\ref{Coeff3a}) and (\ref{Coeff1c}-\ref{Coeff7c}))
we need first to calculate the Weyl symbols of products such as ${X}_{\mu}{X}_{\nu}$,
while from (\ref{Coeff1b}-\ref{Coeff3b}) we need the Weyl symbols of $\mathcal{O}_{1}$, $\mathcal{O}_{1}^{2}$, $\mathcal{O}_{1}{X}_{\mu}$ and $\mathcal{O}_{1}^{2}{X}_{\mu}$. To this end« using the equation (\ref{5.1}) 
it is straightforward to arrive at
\begin{align}\label{Weylsymcons}
\begin{split}
({X}_{\mu})_{W}&=x_{\mu}, \qquad ({X}_{\mu}{X}_{\nu})_{W}=x_{\mu}x_{\nu}+\frac{1}{2}\mathcal{C}_{\mu\nu}^{\sigma}x_{\sigma}, \\
(\mathcal{O}_{1})_{W}&=O_{1},\qquad  (\mathcal{O}_{1}{X}_{\mu})_{W}=O_{1}x_{\mu}, \\
(\mathcal{O}_{1}^{2})_{W}&=O_{1}^{2},\qquad  (\mathcal{O}_{1}^{2}{X}_{\mu})_{W}=
O_{1}^{2}x_{\mu}.
\end{split}
\end{align}  

Moreover, having these results the Weyl symbols for the quantum metric components, in their general form (\ref{center1}), is readily shown to be given by 
\begin{equation}\label{Weylsymbmetr}
({{g}}_{\mu\nu})_{W}=a_{(\mu\nu)}
+a_{(\mu\nu)}^{\sigma}x_{\sigma}
+a_{(\mu\nu)}^{\sigma\tau}(x_{\sigma}x_{\tau}+\frac{1}{2}\mathcal{C}_{\sigma\tau}^{\lambda}x_{\lambda}),
\end{equation}
where it should be noted that both $a_{(\mu\nu)}$ and $a_{(\mu\nu)}^{\sigma}$ are functions of $O_{1}$ while $a_{(\mu\nu)}^{\sigma\tau}$ is independent of it.

Consequently the Weyl symbols of the quantum metric components can be summarized by
\begin{align}\label{Weylmetrcomp}
(\mathfrak{{g}}_{00})_{W}&=a_{0}+a_{1}x_{3}+a_{2}x_{3}^{2}, & (\mathfrak{{g}}_{12})_{W}&=\delta x_{1}x_{2},\nonumber \\
(\mathfrak{{g}}_{01})_{W}&=\alpha x_{1} +2\beta x_{1}x_{3},, & (\mathfrak{{g}}_{13})_{W}&=\alpha^{\prime} x_{1}+2\beta^{\prime} x_{1}x_{3},\nonumber \\
(\mathfrak{{g}}_{02})_{W}&=\alpha x_{2} +2\beta x_{2}x_{3}, & (\mathfrak{{g}}_{22})_{W}&=\delta x_{2}^{2},\\
(\mathfrak{{g}}_{03})_{W}&=b_{0}+b_{1}x_{3}+b_{2}x_{3}^{2}, &
(\mathfrak{{g}}_{23})_{W}&=\alpha^{\prime}x_{2}+2\beta^{\prime}x_{2}x_{3},\nonumber\\
(\mathfrak{{g}}_{11})_{W}&=\delta x_{1}^{2}, &
(\mathfrak{{g}}_{33})_{W}&=c_{0}+c_{1}x_{3}+c_{2}x_{3}^{2}\nonumber.
\end{align}

It is interesting to note also that in these deformed Weyl-metric components the vectors $\xi_{\lambda}$ appear only as quotients  and hence are independent of the inverse Planck length $\kappa$.
To further show that this twisted metric is non-singular we use the quantum determinant expression (\ref{quantumdeter}) and, in order to find the star deformed determinant we replace the metric operators for their corresponding 
Weyl-symbols together with their deformed star multiplication. To this end we first observe that because of the previously derived relations (\ref{consrel1}, \ref{b1}, \ref{b2}, \ref{consrel2}, \ref{consrel3}), together with the additional
relations (\ref{constriction1}) due to the metricity condition  outlined in Sec.4, the quantum determinant (\ref{app}) is considerably simplified to
\be\label{simpqdet}
\det (\frak{{g}}_{\mu\nu})= (b_0-\xi_3 \xi_0 ^{-1} a_0)\Sigma_2 +\xi_3 \xi_0 ^{-1}(b_0-\xi_0 \xi_3 ^{-1} c_0)[\frak{{g}}_{10} J_5 + \frak{{g}}_{20} J_4 ],
\ee

where 
\begin{subequations}
\begin{align}\label{simpqdet2}
&J_4 = -2i\xi_3 \beta\delta X_{1}(3\mathcal O_2 -4X_{1}^2)\\
&J_5 = 2i\xi_3 \beta\delta X_{2}(\mathcal O_2 -4X_{1}^2)\\
&\Sigma _2=(2i\xi_3 \beta\delta) \xi_3 \xi_0 ^{-1}[ 2i\xi_3 \beta \mathcal O_2 ^2 -2 \mathcal O_2 X_{1}X_{2} (\alpha + 2\beta X_{3}) - 4i\xi_3 \beta  \mathcal O_2 X_{1}^2 ]\\
&\xi_3 \xi_0 ^{-1}(\frak{{g}}_{10}J_5 + \frak{{g}}_{20}J_4)=2i\xi_3 \beta\delta (\xi_3 \xi_0 ^{-1})[-10i \xi_3 \beta \mathcal O_2 ^2 -2 \mathcal O_2  X_{1}X_{2} ( \alpha +2\beta X_{3} ) +20i \xi_3 \mathcal O_2 \beta X_{1}^2  ].
\end{align}
\end{subequations}
Substituting these expressions in (\ref{simpqdet}) we obtain 
\be
\det (\frak{{g}}_{\mu\nu})= \xi_3 X_{1}X_{2} K_2 + \xi_3^{2}(K_0 + X_{1}^2 K_1),
\ee

where
\begin{align}
K_0 (\mathcal O_1, \mathcal O_2)=&-4\beta^2 \delta \mathcal O_2 ^2 [5 (c_0 -\xi_3 \xi_0 ^{-1}b_0)+\xi_{3}\xi_{0}^{-1}(b_0 -\xi_3 \xi_0 ^{-1}a_0) ],\\
K_1 (\mathcal O_1, \mathcal O_2)=&8 \beta^2 \delta \mathcal O_2 [5 (c_0 -\xi_3 \xi_0 ^{-1}b_0)+\xi_3 \xi_0 ^{-1} (b_0 -\xi_3 \xi_0 ^{-1}a_0)],\\
K_2 (\mathcal O_1, \mathcal O_2, X_{3})=& 4i\beta\delta \mathcal O_2 [ (c_0 -\xi_3 \xi_0 ^{-1}b_0) -  \xi_3 \xi_0 ^{-1} (b_0 -\xi_3 \xi_0 ^{-1}a_0)(\alpha +2\beta X_{3}) ].
\end{align}

Notice that the result is in agreement with the expression (\ref{QD}) of the Appendix A, where if we put the restriction $a_{1}=\xi_{0}\xi_{3}^{-1}b_{1}=\xi_{0}^{2}\xi_{3}^{-2}c_{1}$ that comes from the metricity, we obtain that these operators $ $ coincide with the $K_{i}$ found there.

The corresponding Weyl functions of the above quantum operators can be now readily obtained by repeated application of the deformation $\star$-product to the products of the algebra.

  The final result of such a procedure is that we get 
\begin{align}\label{appweyl}
\det (\frak{{g}}_{\mu\nu})_{W}=\xi_{3}(x_{1}x_{2}k_{2})+\xi_{3}^{2}(r_{1}+r_{2}x_{1}^{2}),
\end{align}

where $k_{i}$ are the Weyl functions of $K_{i}$ and $r_{1}$, $r_{2}$ are combinations of the $k_{0}$, $k_{1}$ with additional terms that come from their star product with the elements of the algebra,
so that
\begin{align}
k_{2}&=\alpha_{1}+\alpha_{2}x_{3}=4i\beta\delta \mathcal O_2 [ (c_0 -\xi_3 \xi_0 ^{-1}b_0) -  \xi_3 \xi_0 ^{-1} (b_0 -\xi_3 \xi_0 ^{-1}a_0)(\alpha +2\beta x_{3}) ]\\
r_{1}&=k_{0}-i\alpha_{2}O_{2}=-12\beta^{2}\delta O_{2}^{2}[ (c_0 -\xi_3 \xi_0 ^{-1}b_0) +  \xi_3 \xi_0 ^{-1} (b_0 -\xi_3 \xi_0 ^{-1}a_0)]\\
r_{2}&=k_{1}+2i\alpha_{2}=24\beta^{2}\delta O_{2}[ (c_0 -\xi_3 \xi_0 ^{-1}b_0) -  \xi_3 \xi_0 ^{-1} (b_0 -\xi_3 \xi_0 ^{-1}a_0)].
\end{align}

\section{ACKNOWLEDGMENTS}
J. D. V.  acknowledges partial support from CONACYT project 237503 and DGAPA-UNAM grant IN 103716.

\appendix
\section{QUANTUM METRIC DETERMINANT}
In this appendix, we briefly review the explicit calculation of the quantum determinant found at the end of section 3.2 of this paper. For this purpose we will use the relation between the coefficients of the metric given by
\begin{align}\label{A1}
a_{2}&=\xi_{0}\xi_{3}^{-1}b_{2}=\xi_{0}^{2}\xi_{3}^{-2}c_{2},\qquad
\beta^{'}=\xi_{3}\xi_{0}^{-1}\beta,
\end{align}
through these expressions we find that some components of the metric are related to each other in the following way
\begin{align}\label{A2}
\begin{split}
\mathfrak{g}_{03}&=(b_{0}-\xi_{3}\xi_{0}^{-1}a_{0})+(b_{1}-\xi_{3}\xi_{0}^{-1}a_{1})X_{3}+\xi_{3}\xi_{0}^{-1}\mathfrak{g}_{00},\\
\mathfrak{g}_{33}&=(c_{0}-\xi_{3}^{2}\xi_{0}^{-2}a_{0})+(c_{1}-\xi_{3}^{2}\xi_{0}^{-2}a_{1})X_{3}+\xi_{3}^{2}\xi_{0}^{-2}\mathfrak{g}_{00},\\
\mathfrak{g}_{13}&=(\alpha^{\prime}-\xi_{3}\xi_{0}^{-1}\alpha)X_{1}+\xi_{3}\xi_{0}^{-1}\mathfrak{g}_{01},\\
\mathfrak{g}_{23}&=(\alpha^{\prime}-\xi_{3}\xi_{0}^{-1}\alpha)X_{2}+\xi_{3}\xi_{0}^{-1}\mathfrak{g}_{02}.
\end{split}
\end{align}

Now, with the purpose of simplifying the notation we can then express the determinant as:
\be\label{sig3}		 
\det\mathfrak{{g}}_{\mu\nu}={\frak g}_{00}\Sigma_{1}+
{\frak g}_{30}\Sigma_{2}+{\frak g}_{10}\Sigma_{3}-{\frak g}_{20}\Sigma_{4}, 		 
\ee
with the operators $J_{i}$ defined by equations
    \begin{align}\label{A3}
    \begin{split}
	J_{1}&:={\frak g}_{11}{\frak g}_{22}-{\frak g}_{21}{\frak g}_{12}=0,\\
	J_{2}&:={\frak g}_{31}{\frak g}_{12}-{\frak g}_{11}{\frak g}_{32}=-\xi_{3}\xi_{0}^{-1}J_{4},\\
	J_{3}&:={\frak g}_{21}{\frak g}_{32}-{\frak g}_{31}{\frak g}_{22}=-\xi_{0}\xi_{3}^{-1}J_{5},\\
	J_{4}&:={\frak g}_{11}{\frak g}_{02}-{\frak g}_{01}{\frak g}_{12}=2i\xi_{3}\beta\delta X_{1}(4X_{1}^{2}-3\mathcal{O}_{2}),\\
	J_{5}&:={\frak g}_{01}{\frak g}_{22}-{\frak g}_{21}{\frak g}_{02}=2i\xi_{3}\beta\delta X_{2}(\mathcal{O}_{2}-4X_{1}^{2}),\\
	J_{6}&:={\frak g}_{01}{\frak g}_{32}-{\frak g}_{31}{\frak g}_{02}=2i\xi_{3}\beta(\alpha^{'}-\xi_{0}^{-1}\xi_{3}\alpha)(\mathcal{O}_{2}-X_{1}^{2}),
	\end{split}
	\end{align}
where we have used the relations (\ref{A2}) and the algebra to explicitly calculate each of them, whereas that operators $\Sigma_{i}$ are
\begin{align}\label{A4}
\begin{split}
&\Sigma_{1}:=J_{1}{\frak g}_{33}+J_{2}{\frak g}_{23}+J_{3}{\frak g}_{13},\\
&\Sigma_{2}:=J_{4}{\frak g}_{23}+J_{5}{\frak g}_{13}-J_{1}{\frak g}_{03},\\
&\Sigma_{3}:=J_{6}{\frak g}_{23}-J_{5}{\frak g}_{33}-J_{3}{\frak g}_{03},\\
&\Sigma_{4}:=J_{4}{\frak g}_{33}+J_{6}{\frak g}_{13}+J_{2}{\frak g}_{03}.
\end{split}
\end{align}
		 
In order to further calculate these four terms we use las relaciones (\ref{A2}) to obtain
\begin{align}\label{A5}
\begin{split}
&\Sigma_{1}=-\xi_{3}\xi_{0}^{-1}\Sigma_{2},\\
&\Sigma_{2}=(\alpha^{'}-\xi_{3}\xi_{0}^{-1}\alpha)(J_{4}X_{2}+J_{5}X_{1})+\xi_{3}\xi_{0}^{-1}(J_{4}\mathfrak{g}_{02}+J_{5}\mathfrak{g}_{01}),\\
&\Sigma_{3}=\xi_{3}\xi_{0}^{-1}J_{5}\Big((b_{0}-\xi_{0}\xi_{3}^{-1}c_{0})+(b_{1}-\xi_{0}\xi_{3}^{-1}c_{1})X_{3}\Big)+(\alpha^{'}-\xi_{3}\xi_{0}^{-1}\alpha)J_{6}X_{2}+\xi_{3}\xi_{0}^{-1}J_{6}\mathfrak{g}_{02},\\
&\Sigma_{4}=-\xi_{3}\xi_{0}^{-1}J_{4}\Big((b_{0}-\xi_{0}\xi_{3}^{-1}c_{0})+(b_{1}-\xi_{0}\xi_{3}^{-1}c_{1})X_{3}\Big)+(\alpha^{'}-\xi_{3}\xi_{0}^{-1}\alpha)J_{6}X_{1}+\xi_{3}\xi_{0}^{-1}J_{6}\mathfrak{g}_{01},
\end{split}
\end{align}		 
using (\ref{A5}) the determinant becomes
\begin{equation}\label{A6}
\det\mathfrak{g}_{\mu\nu}=\Sigma_{2}R_{1}(X_{3})+(b_{1}-\xi_{3}\xi_{0}^{-1}a_{1})[X_{3},\Sigma_{2}]+(\mathfrak{g}_{10}J_{5}+\mathfrak{g}_{02}J_{4})R_{2}(X_{3})+\mathfrak{g}_{01}J_{6}\mathfrak{g}_{23}-\mathfrak{g}_{02}J_{6}\mathfrak{g}_{13}
\end{equation}
where we have defined the following operators
\begin{align}\label{A7}
\begin{split}
R_{1}(X_{3})&=(b_{0}-\xi_{3}\xi_{0}^{-1}a_{0})+(b_{1}-\xi_{3}\xi_{0}^{-1}a_{1})X_{3},\\
R_{2}(X_{3})&=\xi_{3}\xi_{0}^{-1}\Big[(b_{0}-\xi_{0}\xi_{3}^{-1}c_{0})+(b_{1}-\xi_{0}\xi_{3}^{-1}c_{1})X_{3}\Big].
\end{split}
\end{align}

It is straightforward to find the value of $\Sigma_{2}$, simply by substituting the values of $J_{4}$ and $J_{5}$ obtained in (\ref{A3}) together with the components of the metric $\mathfrak{g}_{01}$ and $\mathfrak{g}_{02}$, while we calculate its commutator with the element $X_{3}$ by using our algebra. In this way it is easy to see that the result is
\begin{align*}
\Sigma_{2}&=-4i\xi_{3}\beta\delta(\alpha^{\prime}-\xi_{3}\xi_{0}^{-1}\alpha)\mathcal{O}_{2}X_{1}X_{2}+4i\xi_3 \beta\delta \mathcal O_2 (\xi_3 \xi_0 ^{-1})\Big( i\xi_3 \beta \mathcal O_2 - X_{1}X_{2} (\alpha + 2\beta X_{3}) - 2i\xi_3 \beta   X_{1}^2 \Big),\\
[X_{3},\Sigma_{2}]&=S_{1}-2X_{1}^{2}S_{1}+X_{1}X_{2}S_{2},
\end{align*}	
with operators $S_{i}$ defined by
\begin{align}\label{A10}
\begin{split}
S_{1}(X_{3})&=8\xi_{3}^{2}\beta\delta\mathcal{O}_{2}^{2}[(\alpha^{\prime}-\xi_{3}\xi_{0}^{-1}\alpha)+\xi_{3}\xi_{0}^{-1}(\alpha+2\beta X_{3})],\\
S_{2}&=32i\xi_{3}^{3}\beta\delta\mathcal{O}_{2}(\xi_{3}\xi_{0}^{-1})
\end{split}.
\end{align}

Thus, by combining these two expressions we can write the first two terms of the determinant (\ref{A6}) as follows
\begin{align}\label{A11}
\Sigma_{2}R_{1}(X_{3})+(b_{1}-\xi_{3}\xi_{0}^{-1}a_{1})[X_{3},\Sigma_{2}]=A_{1}(X_{3})+X_{1}^{2}A_{2}(X_{3})+X_{1}X_{2}A_{3}(X_{3}),
\end{align}
with the operators $A_{i}$ defined by the expressions
\begin{align}\label{A12}
\begin{split}
&A_{1}(X_{3})=-4\xi_{3}^{2}\beta^{2}\delta\mathcal{O}_{2}(\xi_{3}\xi_{0}^{-1})R_{1}(X_{3})+(b_{1}-\xi_{3}\xi_{0}^{-1}a_{1})S_{1}(X_{3}),\\
&A_{2}(X_{3})=8\xi_{3}^{2}\beta^{2}\delta\mathcal{O}_{2}(\xi_{3}\xi_{0}^{-1})R_{1}(X_{3})-2(b_{1}-\xi_{3}\xi_{0}^{-1}a_{1})S_{1}(X_{3}),\\
&A_{3}(X_{3})=-4i\xi_{3}\beta\delta\mathcal{O}_{2}[(\alpha^{'}-\xi_{3}\xi_{0}^{-1}\alpha)+\xi_{3}\xi_{0}^{-1}(\alpha+2\beta X_{3})]R_{1}(X_{3})+(b_{1}-\xi_{3}\xi_{0}^{-1}a_{1})S_{2}.
\end{split}
\end{align}

Now we will focus on calculating the third term in (\ref{A6}), which we can write as
\begin{align}\label{A13}
\mathfrak{g}_{10}J_{5}+\mathfrak{g}_{02}J_{4}=[\mathfrak{g}_{01},J_{5}]+[\mathfrak{g}_{02},J_{4}]+J_{5}\mathfrak{g}_{01}+J_{4}\mathfrak{g}_{02},
\end{align}	
it is easily computed from the relations
\begin{align*}
\begin{split}
[\mathfrak{g}_{10},J_{5}]&=8\xi_{3}^{2}\beta^{2}\delta\mathcal{O}_{2}X_{1}^{2}-32\xi_{3}^{2}\beta^{2}\delta X_{1}^{4}+64\xi_{3}^{2}\beta^{2}\delta X_{1}^{2}X_{2}^{2},\\
[\mathfrak{g}_{02},J_{4}]&=24\xi_{3}^{2}\beta^{2}\delta\mathcal{O}_{2}X_{2}^{2}-96\xi_{3}^{2}\beta^{2}\delta X_{1}^{2}X_{3}^{2},\\
J_{5}\mathfrak{g}_{10}&=2i\xi_{3}\beta\delta\mathcal{O}_{2}X_{1}X_{2}(\alpha+2\beta X_{3})-8i\xi_{3}\beta\delta X_{1}^{3}X_{2}(\alpha+2\beta X_{3})-4\xi_{3}^{2}\beta^{2}\delta\mathcal{O}_{2}X_{2}^{2}+16\xi_{3}^{2}\beta^{2}\delta X_{1}^{2}X_{2}^{2},\\
J_{4}\mathfrak{g}_{02}&=-6i\xi_{3}\beta\delta\mathcal{O}_{2}X_{1}X_{2}(\alpha+2\beta X_{3})+8i\xi_{3}\beta\delta X_{1}^{3}X_{2}(\alpha+2\beta X_{3})-12\xi_{3}^{2}\beta^{2}\delta\mathcal{O}_{2}X_{1}^{2}+16\xi_{3}^{2}\beta^{2}\delta X_{1}^{4}.
\end{split}
\end{align*}	

In this way we have already calculated the third summing of the determinant (\ref{A6}), which we can write as follows
\begin{align}\label{A15}
(\mathfrak{g}_{10}J_{5}+\mathfrak{g}_{02}J_{4})R_{2}(X_{3})=A_{4}(X_{3})+X_{1}^{2}A_{5}(X_{3})+X_{1}X_{2}A_{6}(X_{3})
\end{align}
where, in the above formula
\begin{align}\label{A16}
\begin{split}
A_{4}(X_{3})&=20\xi_{3}^{2}\beta^{2}\delta\mathcal{O}_{2}^{2}R_{2}(X_{3}),\\
A_{5}(X_{3})&=-40\xi_{3}^{2}\beta^{2}\delta\mathcal{O}_{2}R_{2}(X_{3}),\\
A_{6}(X_{3})&=-4i\xi_{3}\beta\delta\mathcal{O}_{2}(\alpha+2\beta X_{3})R_{2}(X_{3}).
\end{split}
\end{align}	

Finally in order to simplify the calculation, we can rewrite the last two terms in (\ref{A6}) as
\begin{align*}
\mathfrak{g}_{01}J_{6}\mathfrak{g}_{23}&=([\mathfrak{g}_{10},J_{6}]+J_{6}\mathfrak{g}_{10})\Big((\alpha^{\prime}-\xi_{3}\xi_{0}^{-1}\alpha)X_{2}+\xi_{3}\xi_{0}^{-1}\mathfrak{g}_{02}\Big)=[\mathfrak{g}_{10},J_{6}]\Big((\alpha^{\prime}-\xi_{3}\xi_{0}^{-1}\alpha)X_{2}+\xi_{3}\xi_{0}^{-1}\mathfrak{g}_{02}\Big)+\\
&+(\alpha^{'}-\xi_{3}\xi_{0}^{-1}\alpha)J_{6}([\mathfrak{g}_{10},X_{2}]+X_{2}\mathfrak{g}_{10})+\xi_{3}\xi_{0}^{-1}J_{6}\mathfrak{g}_{10}\mathfrak{g}_{20},\\
\mathfrak{g}_{02}J_{6}\mathfrak{g}_{13}&=([\mathfrak{g}_{20},J_{6}]+J_{6}\mathfrak{g}_{20})\Big((\alpha^{\prime}-\xi_{3}\xi_{0}^{-1}\alpha)X_{1}+\xi_{3}\xi_{0}^{-1}\mathfrak{g}_{01}\Big)=[\mathfrak{g}_{20},J_{6}]\Big((\alpha^{\prime}-\xi_{3}\xi_{0}^{-1}\alpha)X_{1}+\xi_{3}\xi_{0}^{-1}\mathfrak{g}_{01}\Big)+\\
&+(\alpha^{\prime}-\xi_{3}\xi_{0}^{-1}\alpha)J_{6}([\mathfrak{g}_{20},X_{1}]+X_{1}\mathfrak{g}_{20})+\xi_{3}\xi_{0}^{-1}J_{6}\mathfrak{g}_{20}\mathfrak{g}_{10},
\end{align*}
thereby we have that its difference is
\begin{align}\nonumber
\mathfrak{g}_{01}J_{6}\mathfrak{g}_{23}-\mathfrak{g}_{02}J_{6}\mathfrak{g}_{13}
&=[\mathfrak{g}_{10},J_{6}]\Big((\alpha^{\prime}-\xi_{3}\xi_{0}^{-1}\alpha)X_{2}+\xi_{3}\xi_{0}^{-1}\mathfrak{g}_{02}\Big)-[\mathfrak{g}_{20},J_{6}]\Big((\alpha^{\prime}-\xi_{3}\xi_{0}^{-1}\alpha)X_{1}+\xi_{3}\xi_{0}^{-1}\mathfrak{g}_{01}\Big)\\\label{A17}
&+(\alpha^{\prime}-\xi_{3}\xi_{0}^{-1}\alpha)J_{6}([\mathfrak{g}_{10},X_{2}]-[\mathfrak{g}_{20},X_{1}]+X_{2}\mathfrak{g}_{10}-X_{1}\mathfrak{g}_{20})+\xi_{3}\xi_{0}^{-1}J_{6}[\mathfrak{g}_{10},\mathfrak{g}_{20}].
\end{align}

Moreover, one can check by straightforward calculations that the commutators in the above expression are
\begin{align}\label{A18}
\begin{split}
[\mathfrak{g}_{10},J_{6}]&=16\xi_{3}^{2}\beta^{2}(\alpha^{\prime}-\xi_{3}\xi_{0}^{-1}\alpha)X_{1}^{2}X_{2},\\
[\mathfrak{g}_{20},J_{6}]&=16\xi_{3}^{2}\beta^{2}(\alpha^{\prime}-\xi_{3}\xi_{0}^{-1}\alpha)X_{1}X_{2}^{2},\\
[\mathfrak{g}_{10},X_{3}]&=-4i\xi_{3}\beta X_{1}^{2},\\
[\mathfrak{g}_{20},X_{1}]&=4i\xi_{3}\beta X_{2}^{2},\\
[\mathfrak{g}_{10},\mathfrak{g}_{20}]&=-4i\xi_{3}\beta\mathcal{O}_{2}(\alpha+2\beta X_{3}),
\end{split}
\end{align}
substituting this into (\ref{A17}) we obtain that this expression has the simple form
\begin{align}\label{A19}
\mathfrak{g}_{01}J_{6}\mathfrak{g}_{23}-\mathfrak{g}_{02}J_{6}\mathfrak{g}_{13}=A_{7}(X_{3})-X_{1}^{2}A_{7}(X_{3})+X_{1}X_{2}A_{8},
\end{align}
where their coefficients are given by the following expressions
\begin{align}\label{A20}
\begin{split}
A_{7}(X_{3})&=4\xi_{3}^{2}\beta^{2}\mathcal{O}_{2}^{2}(\alpha^{\prime}-\xi_{3}\xi_{0}^{-1}\alpha)[(\alpha^{\prime}-\xi_{3}\xi_{0}^{-1}\alpha)+2\xi_{3}\xi_{0}^{-1}(\alpha+2\beta X_{3})],\\
A_{8}&=-32i\xi_{3}^{3}\beta^{3}\mathcal{O}_{2}(\alpha^{\prime}-\xi_{3}\xi_{0}^{-1}\alpha)\xi_{3}\xi_{0}^{-1}.
\end{split}
\end{align}	

Therefore, by substituting the expressions (\ref{A11}), (\ref{A15}) and (\ref{A19}) in (\ref{A6}) we obtain that the quantum determinant reads
\begin{align}\label{QD}
\det\mathfrak{g}_{\mu\nu}=K_{0}(X_{3})+X_{1}^{2}K_{1}(X_{3})+X_{1}X_{2}K_{2}(X_{3}),
\end{align}
where, the operators $K_{\mu}$ are defined in terms of the $A_{i}$ as follows
\begin{align}\label{A22}
\begin{split}
K_{0}(X_{3})&=A_{1}(X_{3},\xi_{3}^{2})+A_{4}(X_{3},\xi_{3}^{2})+A_{7}(X_{3},\xi_{3}^{2}),\\
K_{1}(X_{3})&=A_{2}(X_{3},\xi_{3}^{2})+A_{5}(X_{3},\xi_{3}^{2})-A_{7}(X_{3},\xi_{3}^{2}),\\
K_{2}(X_{3})&=A_{3}(X_{3},\xi_{3}^{3})+A_{6}(X_{3},\xi_{3})+A_{8}(\xi_{3}^{3});
\end{split}
\end{align}
we note that in these expressions $K_{i}=K_{i}(\mathcal{Z(A)},X_{3})$ for $i=0,1,2$. 

Also note that this determinant is expressed in terms of powers of $\xi_{3}$, so one could calculate in principle
the inverse of $\mathfrak{g}_{\mu\nu}$ to some order in $\xi_{3}$.  On the other hand, since $\xi_{3}$ is the only element in the  above determinant that depends on the Planck length, setting it equal to zero would make 
(\ref{QD})  singular. It would appear at first sight that the limit of the metric, when (\ref{A12}), (\ref{A16}), (\ref{A17}) and (\ref{A22}) $\to 0$,  is the classical metric of our system. This is of course obviously wrong, since  
the classical limit $\kappa\to 0$, in (\ref{opdef1}-\ref{const1}) would lead to a total commutativity of the generators, {\it i.e.}  not a Lie-algebra and no Inner derivations. Therefore no centrality conditions at all and the problem would then be reduced to one of standard differential geometry.  

\section{METRICITY CONDITION}
Beginning with the metricity condition: \\
\be D_{{X}_{\rho}}\mathfrak{{g}}_{\mu\nu}=\Gamma_{\rho\mu}^{\sigma}\mathfrak{{g}}_{\sigma\nu}+\Gamma_{\rho\nu}^{\sigma}\mathfrak{{g}}_{\sigma\mu},\label{metricitycondappendix}
\end{equation}
with $\Gamma_{\mu\nu}^{\gamma}\in\mathcal{Z(A)}$ as given in Section 4,
we illustrate here with some explicit calculations the results obtained there to derive the relations for the connection symbols in order to satisfy the conditions of metricity and zero torsion. Consider first the case when  $\mu=\nu$ so the equation (\ref{metricitycondappendix}) becomes
\begin{equation}\label{metricitycondappendix-3}
D_{{X}_{\rho}}\mathfrak{{g}}_{\mu\mu}=2\Gamma_{\rho\mu}^{\sigma}\mathfrak{{g}}_{\sigma\mu}.
\end{equation} 

Let us start from $\mu=\nu=1$ the equation (\ref{metricitycondappendix-3}) implies that 
\begin{align*}
D_{{X}_{\rho}}\mathfrak{{g}}_{11}=2\Gamma_{\rho1}^{\sigma}\mathfrak{{g}}_{\sigma1},
\end{align*}
we can calculate both sides in the last equation by means of the relations (\ref{Coeff1c}-\ref{Coeff7c}) and (\ref{der8}), we thus have
\begin{align*}
\delta \mathcal{N}_{\rho1}^{1}X_{1}^{2}+\delta \mathcal{N}_{\rho1}^{2}X_{1}X_{3}=(\alpha\Gamma_{\rho1}^{0}+\alpha^{\prime}\Gamma_{\rho1}^{3})X_{1}+\delta\Gamma_{\rho1}^{1}X_{1}^{2}+\delta\Gamma_{\rho1}^{2}X_{1}X_{3}+(\beta\Gamma_{\rho1}^{0}+\beta^{\prime}\Gamma_{\rho1}^{3})(X_{1}X_{3}+X_{3}X_{1}),
\end{align*}
in order to this equation holds we need that $a_{1}=\xi_{0}\xi_{3}^{-1}b_{1}=\xi_{0}^{2}\xi_{3}^{-2}c_{1}$ and
\begin{align}\label{conditions4}
\Gamma_{\rho1}^{3}=-\xi_{0}\xi_{3}^{-1}\Gamma_{\rho1}^{0},\qquad 
\Gamma_{\rho1}^{1}=\mathcal{N}_{\rho1}^{1},\qquad \Gamma_{\rho1}^{2}=\mathcal{N}_{\rho1}^{2}.
\end{align}
In the same way we obtain for $\mu=\nu=2$
\begin{equation}\label{conditions5}	
\Gamma_{\rho2}^{3}=-\xi_{0}\xi_{3}^{-1}\Gamma_{\rho2}^{0},\qquad \Gamma_{\rho2}^{1}=-\mathcal{N}_{\rho1}^{2},\qquad \Gamma_{\rho2}^{2}=\mathcal{N}_{\rho1}^{1},
\end{equation}
and in a similar manner, resorting to (\ref{Coeff1a}) and (\ref{der8}) for $\mu=\nu=0$,  the left side of (\ref{metricitycondappendix-3}) yields
\begin{align*}
D_{{X}_{\rho}}\mathfrak{{g}}_{00}&=D_{{X}_{\rho}}(a_{0}+a_{1}X_{3}+a_{2}X_{3}^{2})\\
&=D_{X_{\rho}}a_{0}+(D_{X_{\rho}}a_{1})X_{3}+a_{1}(D_{X_{\rho}}X_{3})+a_{2}X_{3} (D_{X_{\rho}}X_{3}) +a_{2}(D_{X_{\rho}}X_{3})X_{3}\\
&=(D_{X_{\rho}}a_{0}+a_{1}\xi_{3}^{-1}\mathcal{N}_{\rho3}^{0}\mathcal{O}_{1})+
a_{1}\xi_{3}\xi_{0}^{-1}\mathcal{N}_{\rho0}^{1}X_{1}+
a_{1}\xi_{3}\xi_{0}^{-1}\mathcal{N}_{\rho0}^{2}X_{3}+
(D_{X_{\rho}}a_{1}+2a_{2}\xi_{3}^{-1}\mathcal{N}_{\rho3}^{0}\mathcal{O}_{1})X_{3}\\
&+a_{2}\xi_{3}\xi_{0}^{-1}\mathcal{N}_{\rho0}^{1}(X_{1}X_{3}+X_{3}X_{1})+a_{2}\xi_{3}\xi_{0}^{-1}\mathcal{N}_{\rho0}^{2}(X_{3}X_{3}+X_{3}X_{3}),
\end{align*}
while, on the other hand, using (\ref{Coeff1a}), (\ref{Coeff1b}), (\ref{Coeff1c}) and (\ref{Coeff2c}), the right side can be written as
\begin{align*}
2\Gamma_{\rho0}^{\sigma}\mathfrak{{g}}_{\sigma0}&=2(a_{0}\Gamma_{\rho0}^{0}+b_{0}\Gamma_{\rho0}^{3})+2\alpha\Gamma_{\rho0}^{1}X_{1}+2\alpha\Gamma_{\rho0}^{2}X_{3}+2(a_{1}\Gamma_{\rho0}^{0}+b_{1}\Gamma_{\rho0}^{3})X_{3}+2(a_{2}\Gamma_{\rho0}^{0}+b_{2}\Gamma_{\rho0}^{3})X_{3}^{2}\\
&+2\beta\Gamma_{\rho0}^{1}(X_{1}X_{3}+X_{3}X_{1})+2\beta\Gamma_{\rho0}^{2}(X_{3}X_{3}+X_{3}X_{3}).
\end{align*}
We thus arrive at the following relations
\begin{subequations}
\begin{align}\label{condition6}
&\Gamma_{\rho0}^{1}=\mathcal{N}_{\rho0}^{1},\qquad \Gamma_{\rho0}^{2}=\mathcal{N}_{\rho0}^{2},\\\label{condition6-1}
&\Gamma_{\rho0}^{3}=-\xi_{0}\xi_{3}^{-1}\Gamma_{\rho0}^{0},\\\label{condition7}
&D_{X_{\rho}}a_{1}+2a_{2}\xi_{3}^{-1}\mathcal{N}_{\rho3}^{0}\mathcal{O}_{1}=0,\\\label{condition7-1}
&2(a_{0}-b_{0}\xi_{0}\xi_{3}^{-1})\Gamma_{\rho0}^{0}=D_{X_{\rho}}a_{0}+a_{1}\xi_{3}^{-1}\mathcal{N}_{\rho3}^{0}\mathcal{O}_{1}.
\end{align}
\end{subequations}
To deal with the case $\mu=\nu=3$ we use an approach analogous that the one described above and obtain the following conditions:
\begin{subequations}
\begin{align}\label{condition8}
&\Gamma_{\rho3}^{1}=\xi_{3}\xi_{0}^{-1}\mathcal{N}_{\rho0}^{1},\qquad \Gamma_{\rho3}^{2}=\mathcal{N}_{\rho0}^{2},\\\label{condition8-1}
&\Gamma_{\rho3}^{3}=-\xi_{0}\xi_{3}^{-1}\Gamma_{\rho3}^{0},\\\label{condition9}
&D_{X_{\rho}}c_{1}+2c_{2}\xi_{3}^{-1}\mathcal{N}_{\rho3}^{0}\mathcal{O}_{1}=0,\\\label{condition9-1} 
&2(b_{0}-c_{0}\xi_{0}\xi_{3}^{-1})\Gamma_{\rho3}^{0}=D_{X_{\rho}}c_{0}+
c_{1}\xi_{3}^{-1}\mathcal{N}_{\rho3}^{0}\mathcal{O}_{1}.
\end{align}
\end{subequations}
Now, for the case $\mu=0$, $\nu=1$ in (\ref{metricitycondappendix}), we can calculate the left side explicitly by making use of (\ref{Coeff1c}), (\ref{der8}) and (\ref{derx0-1}). Thus we get
\begin{align*}
D_{{X}_{\rho}}\mathfrak{{g}}_{01}&=D_{{X}_{\rho}}[\alpha X_{1}+\beta(X_{1}X_{3}+X_{3}X_{1})]\\
&=(D_{X_{\rho}}\alpha)X_{1}+\alpha (D_{{X}_{\rho}}{X_{1}})+\beta[(D_{{X}_{\rho}}X_{3})
X_{1}+X_{3}(D_{{X}_{\rho}}X_{1})+(D_{{X}_{\rho}}X_{1})
X_{3}+X_{1}(D_{{X}_{\rho}}X_{3})]\\
&=(D_{X_{\rho}}\alpha+\alpha \mathcal{N}_{\rho1}^{1}+2\beta\xi_{3}^{-1}\mathcal{N}_{\rho3}^{0}\mathcal{O}_{1})X_{1}+\alpha \mathcal{N}_{\rho1}^{2}X_{3}+2\beta\xi_{3}\xi_{0}^{-1}\mathcal{N}_{\rho0}^{1}X_{1}^{2}+\beta \mathcal{N}_{\rho1}^{1}(X_{1}X_{3}+X_{3}X_{1})\\
&+\beta \mathcal{N}_{\rho1}^{\;\;\;2}(X_{3}X_{3}+X_{3}X_{3})+2\beta\xi_{3}\xi_{0}^{-1}\mathcal{N}_{\rho0}^{2}X_{1}X_{3},
\end{align*}
while the right side yields
\begin{align*}
\Gamma_{\rho0}^{\sigma}\mathfrak{{g}}_{\sigma 1}+\Gamma_{\rho 1}^{\sigma}\mathfrak{{g}}_{\sigma0}&=\Gamma_{\rho0}^{0}\mathfrak{{g}}_{01}+\Gamma_{\rho0}^{1}\mathfrak{{g}}_{11}+\Gamma_{\rho0}^{2}\mathfrak{{g}}_{21}+\Gamma_{\rho0}^{3}\mathfrak{{g}}_{31}+
\Gamma_{\rho1}^{0}\mathfrak{{g}}_{00}+
\Gamma_{\rho1}^{1}\mathfrak{{g}}_{10}+
\Gamma_{\rho1}^{\;\;\;2}\mathfrak{{g}}_{20}+\Gamma_{\rho1}^{3}\mathfrak{{g}}_{30}\\
&=(a_{0}\Gamma_{\rho1}^{0}+b_{0}\Gamma_{\rho1}^{3})+\alpha(\Gamma_{\rho0}^{0}+\Gamma_{\rho1}^{1}+\Gamma_{\rho0}^{3})X_{1}+\alpha\Gamma_{\rho1}^{2}X_{3}+(a_{1}\Gamma_{\rho1}^{0}+b_{1}\Gamma_{\rho1}^{3})X_{3}\\
&+\delta\Gamma_{\rho0}^{1}X_{1}^{2}+\Gamma_{\rho0}^{2}X_{1}X_{3}+(a_{2}\Gamma_{\rho1}^{0}+b_{2}\Gamma_{\rho1}^{3})X_{3}^{2}+\beta(\Gamma_{\rho0}^{0}+\Gamma_{\rho1}^{1}+\Gamma_{\rho0}^{3})(X_{1}X_{3}+X_{3}X_{1})\\
&+\beta\Gamma_{\rho1}^{2}(X_{3}X_{3}+X_{3}X_{3}).
\end{align*}
When  matching both sides and using the expressions in (\ref{conditions4}) we obtain the following restriction
\begin{align}\label{condition1}
(a_{0}-\xi_{0}\xi_{3}^{-1}b_{0})\Gamma_{\rho1}^{0}&=0.
\end{align}
Finally appling the same procedure for $\mu=1, \nu=3$ where the derivation of $\mathfrak{{g}}_{13}$ can be written as
\begin{align}\nonumber
D_{{X}_{\rho}}\mathfrak{{g}}_{13}&=\xi_{3}\xi_{0}^{-1}(D_{X_{\rho}}\alpha+\alpha \mathcal{N}_{\rho1}^{1}+2\beta\xi_{3}^{-1}\mathcal{N}_{\rho3}^{0}\mathcal{O}_{1})X_{1}+\xi_{3}\xi_{0}^{-1}\alpha \mathcal{N}_{\rho1}^{2}X_{3}+2\beta\xi_{3}^{2}\xi_{0}^{-2}\mathcal{N}_{\rho0}^{1}X_{1}^{2}\\\nonumber
&+\xi_{3}\xi_{0}^{-1}\beta \mathcal{N}_{\rho1}^{2}(X_{3}X_{3}+X_{3}X_{3})+\xi_{3}\xi_{0}^{-1}\beta \mathcal{N}_{\rho1}^{1}(X_{1}X_{3}+X_{3}X_{1})+2\beta\xi_{3}^{2}\xi_{0}^{-2}\mathcal{N}_{\rho0}^{2}X_{1}X_{3}\\\nonumber
&=(b_{0}\Gamma_{\rho1}^{0}+c_{0}\Gamma_{\rho1}^{3})+\xi_{3}\xi_{0}^{-1}\alpha(\Gamma_{\rho1}^{1}+\Gamma_{\rho3}^{3}+\xi_{0}\xi_{3}^{-1}\Gamma_{\rho3}^{0})X_{1}+\xi_{3}\xi_{0}^{-1}\alpha\Gamma_{\rho1}^{2}X_{3}+(b_{1}\Gamma_{\rho1}^{0}+c_{1}\Gamma_{\rho1}^{3})X_{3}\\\nonumber
&+(b_{2}\Gamma_{\rho1}^{0}+c_{2}\Gamma_{\rho1}^{3})X_{3}^{2}+\delta\Gamma_{\rho3}^{1}X_{1}^{2}+\delta\Gamma_{\rho3}^{2}X_{1}X_{3}+\xi_{3}\xi_{0}^{-1}\beta(\Gamma_{\rho1}^{1}+\Gamma_{\rho3}^{3}+\xi_{0}\xi_{3}^{-1}\Gamma_{\rho3}^{0})(X_{1}X_{3}+X_{3}X_{1})\\\nonumber
&+\xi_{3}\xi_{0}^{-1}\beta\;\Gamma_{\rho1}^{2}(X_{3}X_{3}+X_{3}X_{3}),
\end{align}
and making use of the coefficient relations in equations (\ref{conditions4}), (\ref{conditions5}),  we get the condition
\begin{equation}\label{condition2}
(b_{0}-\xi_{0}\xi_{3}^{-1}c_{0})\Gamma_{\rho1}^{0}=0.
\end{equation} 
Furthermore, following a procedure analogous to that described above for $\mu=0,\nu=2$ and $\mu=2,\nu=3$ we obtain the last two constraints
\begin{subequations}
\begin{align}\label{condition3}
(a_{0}-\xi_{0}\xi_{3}^{-1}b_{0})\Gamma_{\rho2}^{0}=0,\\\label{condition4}
(b_{0}-\xi_{0}\xi_{3}^{-1}c_{0})\Gamma_{\rho2}^{0}=0.
\end{align}
\end{subequations}
The remaining indices in (\ref{metricitycondappendix}) do not provide new constraints.

Note now that we have two possible solutions for the equations (\ref{condition1}-\ref{condition4}), one of them involves taking $a_{0}=\xi_{0}\xi_{3}^{-1}b_{0}=\xi_{0}^{2}\xi_{3}^{-2}c_{0}$. This however we discard because it will imply metric components such that the quantum determinant (\ref{QD}) would become zero. Thus the only admissible solution is $\Gamma_{\rho1}^{0}=\Gamma_{\rho2}^{0}=0$.

On the other hand, from (\ref{condition6-1}-\ref{condition7-1}) and (\ref{condition8-1}-\ref{condition9-1}), further relations between the remaining $\Gamma$'s must exist. Indeed using the equations (\ref{condition7}) (or (\ref{condition9}) which are the same because of the condition $a_{i}=\xi_{0}\xi_{3}^{-1}b_{i}=\xi_{0}^{2}\xi_{3}^{2}c_{i}$), together with equation (\ref{4.23b}), which relates the $\phi_{\rho}$ function to the $\mathcal{N}$'s, and the expressions for the coefficients of the metric (\ref{Coeff1b}-\ref{Coeff3b}) as well as the derivation of the central element (\ref{bilin4}) to obtain the following relations 
\begin{subequations}
\begin{align}\label{B10a}
\mathcal{N}_{\rho3}^{0}&=l(\xi_{0}\xi_{3}^{-1}l-2a_{2})^{-1}\mathcal{N}_{\rho0}^{0},\\\label{B10b}
\phi_{\rho}&=l_{1}\mathcal{N}_{\rho0}^{0},
\end{align}
\end{subequations}
and using these equations in addition to the (\ref{condition6-1}), (\ref{condition7-1}), (\ref{condition8-1}) and (\ref{condition9-1}) we arrive at
\begin{subequations}
\begin{align*}
\Gamma_{\rho0}^{0}&=l_{2}(\mathcal{O}_{1})\mathcal{O}_{1}\mathcal{N}_{\rho0}^{0}, &  & \Gamma_{\rho3}^{0}=l_{3}(\mathcal{O}_{1})\mathcal{O}_{1}\mathcal{N}_{\rho0}^{0}\\
\Gamma_{\rho0}^{3}&=-\xi_{0}\xi_{3}^{-1}l_{2}(\mathcal{O}_{1})\mathcal{O}_{1}\mathcal{N}_{\rho0}^{0}, & & \Gamma_{\rho3}^{3}=-\xi_{0}\xi_{3}^{-1}l_{3}(\mathcal{O}_{1})\mathcal{O}_{1}\mathcal{N}_{\rho0}^{0},
\end{align*}
\end{subequations}
where we have defined
\begin{align*}
l&=\gamma_{3}+2\xi_{0}\xi_{3}^{-1}\gamma_{4},\\
l_{1}&=2a_{2}(2a_{2}-\xi_{0}\xi_{3}^{-1}l)^{-1},\\
l_{2}(\mathcal{O}_{1})&=\frac{1}{2}(\xi_{3}a_{0}-b_{0}\xi_{0})^{-1}[(\gamma_{1}+2\gamma_{4}\xi_{3}^{-1}\mathcal{O}_{1})l_{1}+a_{1}l(\xi_{0}\xi_{3}^{-1}l-2a_{2})^{-1}],\\
l_{3}(\mathcal{O}_{1})&=\frac{1}{2}(\xi_{3}b_{0}-c_{0}\xi_{0})^{-1}[(\kappa_{1}+2\kappa_{4}\xi_{3}^{-1}\mathcal{O}_{1})l_{1}+\xi_{0}^{-2}\xi_{3}^{2}a_{1}l(\xi_{0}\xi_{3}^{-1}l-2a_{2})^{-1}].
\end{align*}
It therefore follows readily that by using the conditions
\begin{subequations}
\begin{align}\label{gammas1}
&\Gamma_{\rho1}^{0}=\Gamma_{\rho2}^{0}=0,\\\label{gammas2}
&\Gamma_{\rho1}^{1}=\Gamma_{\rho2}^{2}=\mathcal{N}_{\rho1}^{1},\\\label{gammas3}
&\Gamma_{\rho3}^{1}=\xi_{3}\xi_{0}^{-1}\Gamma_{\rho0}^{1}=\xi_{3}\xi_{0}^{-1}\mathcal{N}_{\rho0}^{1},\\\label{gammas4}
&\Gamma_{\rho1}^{2}=-\Gamma_{\rho2}^{1}=\mathcal{N}_{\rho1}^{2},\\\label{gammas5}
&\Gamma_{\rho3}^{2}=\xi_{3}\xi_{0}^{-1}\Gamma_{\rho0}^{2}=\xi_{3}\xi_{0}^{-1}\mathcal{N}_{\rho0}^{2},\\\label{gammas6}
&\Gamma_{\rho0}^{3}=-\xi_{0}\xi_{3}^{-1}\Gamma_{\rho0}^{0}=-\xi_{0}\xi_{3}^{-1}l_{2}\mathcal{O}_{1}\mathcal{N}_{\rho0}^{0},\\\label{gammas7}
&\Gamma_{\rho1}^{3}=\Gamma_{\rho2}^{3}=0,\\\label{gammas8}
&\Gamma_{\rho3}^{3}=-\xi_{0}\xi_{3}^{-1}\Gamma_{\rho3}^{0}=-\xi_{0}\xi_{3}^{-1}l_{3}\mathcal{O}_{1}\mathcal{N}_{\rho0}^{0},
\end{align}
\end{subequations}
the metricity condition is completely satisfied with these connection symbols valued in
$\mathcal Z(A)$.

\end{document}